%% file: main.tex
\documentclass{article}
\usepackage{url}
\usepackage{booktabs} 
\usepackage[ruled]{algorithm2e} 

\input{includes}
\input{defines}
\usepackage{pgfplots}
\pgfplotsset{width=7.0cm, compat=1.9}
\usepackage{tikz}
\usepackage{subcaption}
\usepackage{booktabs}

\input{title}

\begin{document}

\maketitle

\input{abstract}

\input{intro}

\input{preliminaries}
\input{def-spam-resistance}
\input{def-distortion}
\input{results}
\input{closure}

\input{related}

\bibliographystyle{plain}
\bibliography{pagerank,dblp-conferences,dblp-papers}

\section*{Appendix}
\appendix
\input{distortion}
\input{new-spam-resistance}
\input{experiments}

\end{document}

%% file: includes.tex
\usepackage{xspace}
\usepackage{times}
\usepackage{enumerate}
\usepackage{comment}
\usepackage{our-comments}
\usepackage{url}
\usepackage{amsmath}
\usepackage{amsthm}
\usepackage{amssymb}
\usepackage{pdfpages}

\usepackage{paralist}
\usepackage[active]{srcltx}
\usepackage{fullpage}
\usepackage{color}
\usepackage{soul}
\usepackage[font=normalsize]{caption}
\usepackage[export]{adjustbox}

\usepackage{hyperref}        
\usepackage{cleveref}

\usepackage[normalem]{ulem} 

\usepackage{graphicx}
\usepackage{psfrag}       

\usepackage[shortlabels]{enumitem}

\usepackage{relate}

%% file: defines.tex
\newcommand{\set}[1]{\mathcal{#1}}
\newcommand{\vect}[1]{\mathbf{\hat{#1}}}
\newcommand{\scaledvect}[1]{\mathbf{#1}}

\newcommand{\mtermi}[1]{\mathcal{R}_{\textnormal{min}}(#1)}
\newcommand{\mtermialgo}[2]{\textnormal{Min-PPR}_{#1,#2}}
\newcommand{\mtermialgot}[2]{\textnormal{T-Min-PPR}_{#1,#2}}

\newcommand{\medtermi}[1]{\mathcal{R}_{\textnormal{med}}(#1)}

\newcommand{\graph}[1]{\mathbf{#1}}
\newcommand{\mtrx}[1]{\mathbf{#1}}

\newcommand{\eps}{\varepsilon}
\let\epsilon=\varepsilon 

\newcommand{\nodes}{\mathbf{V}}
\newcommand{\edges}{\mathbf{E}}
\newcommand{\trust}{\mathbf{T}}
\newcommand{\centers}{\mathbf{K}}
\newcommand{\prset}{\mathbf{X}}

\newcommand{\spam}{\mathbf{S}}
\newcommand{\purchase}{\mathbf{P}}
\newcommand{\Pgraphs}{\graph{G}_{\purchase}}

\newcommand{\rp}{\epsilon} 
\newcommand{\R}{\vect{r}} 

\newcommand{\calE}{\mathcal{E}}
\newcommand{\calG}{\mathcal{G}} 

\newcommand{\norm}[1]{\mathcal{N}(#1)}

\newcommand{\setP}{\set{P}}
\newcommand{\p}{\vect{p}}
\newcommand{\x}{\vect{x}}

\newcommand{\OUT}{N_\textrm{out}}
\newcommand{\IN}{N_\textrm{in}}
\newcommand{\din}{d_\textrm{in}}
\newcommand{\dout}{d_\textrm{out}}

\newcommand{\mam}[1]{{\color{red}{#1}}}

\newtheorem{theorem}{Theorem}

\newtheorem{lemma}[theorem]{Lemma}

\newtheorem{definition}[theorem]{Definition}

\newtheorem{observation}[theorem]{Observation}

\newcommand{\pr}{PageRank\xspace}
\newcommand{\prs}{PageRanks\xspace}
\newcommand{\PPR}{PPR\xspace}

\newcommand{\PPRs}{PPRs\xspace}
\newcommand{\UPR}{UPR\xspace}

\newcommand{\mTPPR}{\mbox{T-PPR}_\rp}

\newcommand{\minppr}{Min-PPR\xspace}

\newcommand{\tminppr}{T-\minppr}

\def\hyperspc{\kern -0.22em}

\newcommand{\rhfloor}{\rfloor \hyperspc \rfloor}
\newcommand{\lhceil}{\lceil \hyperspc \lceil}

\renewcommand{\norm}[1]{\lhceil #1 \rhfloor}

\renewcommand*\bar[1]{%
   \hbox{%
     \vbox{%
       \hrule height 0.8pt 
       \kern0.2ex
       \hbox{%
         \kern 0em
         \ensuremath{#1}%
         \kern 0.1em
       }%
     }%
   }%
}

\newcommand{\pp}[1]{\medskip\noindent\textbf{#1}}

\newcommand{\station}[1]{\mathcal{R}(#1)}
\newcommand{\stationinv}[1]{\mathcal{R}^{-1}(#1)}
\newcommand{\stationmin}[1]{\mathcal{R}_{\min}(#1)}

\newcommand{\uniform}{\vect{u}}

\newcommand{\distortion}{\mathbf{D}}
\newcommand{\ergodicG}{\mathcal{G}}

\newcommand{\str}{\textnormal{Stretch}} 
\newcommand{\cont}{\textnormal{Cont}}   
\newcommand{\dtv}{d_{\textnormal{TV}}}  

\newcommand{\floor}[1]{\lfloor #1 \rfloor}
\newcommand{\ceil}[1]{\lceil #1 \rceil}
\newcommand{\mean}{\mathbb{E}}
\newcommand{\dist}{\delta}
\newcommand{\prob}{\mathbb{P}}

\newcommand{\reference}{reference\xspace}

\newcommand{\mPR}[1]{\station{\graph{G},#1,\rp}}

\newcommand{\sr}{\sigma}

\newcommand{\mix}{\tau} 
\newcommand{\mixclass}[1]{\ergodicG_{#1}} 
\newcommand{\emix}{\textnormal{emt}} 

\newcommand{\mymod}{\,\mathrm{mod}\,}

%% file: title.tex
\title{Graph Ranking and the Cost of Sybil Defense}

\author{
Gwendolyn Farach-Colton\\
Rutgers University\\
\href{mailto:gwen@farach-colton.com}{gwen@farach-colton.com}
\and 
Mart\'{\i}n Farach-Colton\\
Rutgers University\\
\href{mailto:martin@farach-colton.com}{martin@farach-colton.com}
\and 
Leslie Ann Goldberg\\
Oxford University\\
\href{mailto:leslie.goldberg@cs.ox.ac.uk}{leslie.goldberg@cs.ox.ac.uk}
\and 
Hanna Koml\'os\\
Rutgers University\\
\href{mailto:hkomlos@gmail.com}{hkomlos@gmail.com}
\and 
John Lapinskas\\
Bristol University\\
\href{mailto:john.lapinskas@bristol.ac.uk}{john.lapinskas@bristol.ac.uk}
\and 
Reut Levi\\
Reichman University\\
\href{mailto:reut.levi1@idc.ac.il}{reut.levi1@idc.ac.il}
\and 
Moti Medina\\
Bar-Ilan University\\
\href{mailto:moti.medina@biu.ac.il}{moti.medina@biu.ac.il}
\and 
Miguel A. Mosteiro\\
Pace University\\
\href{mailto:mmosteiro@pace.edu}{mmosteiro@pace.edu}
}

%% file: abstract.tex
\begin{abstract}

Ranking functions such as \pr assign numeric values (\emph{ranks}) to nodes of graphs, most notably the web graph.  Node rankings are an integral part of Internet search algorithms, since they can be used to order the results of queries. However, these ranking functions are famously subject to attacks by spammers, who modify the web graph in order to give their own pages more rank. 


We characterize the interplay between rankers and spammers as a game.  We define the two critical features of this game, \emph{spam resistance} and \emph{distortion}, based on how spammers spam and how rankers protect against spam.  We observe that all the ranking functions that are well-studied in the literature, including the original formulation of \pr,  have poor spam resistance, poor distortion, or both.



Finally, we study Min-PPR, the form of PageRank used at Google itself, but which has received no (theoretical or empirical) treatment in the literature.  We prove that Min-PPR has low distortion and high spam resistance. A secondary benefit is that Min-PPR comes with an explicit cost function on nodes that shows how important they are to the spammer; thus a ranker can focus their spam-detection capacity on these vulnerable nodes.  Both Min-PPR and its associated cost function are straightforward to compute.

\end{abstract}

%% file: intro.tex


\section{Introduction}\label{sec:intro}



Ranking functions such as \pr~\cite{DBLP:journals/cn/BrinP98} assign numeric values (\emph{ranks}) to nodes of graphs, most notably the web graph.  Node rankings are an integral part of Internet search algorithms, since they can be used to order the results of queries. 
However, these ranking functions are famously subject to attacks by spammers, who modify the web graph in order to give their own pages more rank. 

In the literature on ranking functions, resistance to spam attacks is treated heuristically~\cite{fogaras2005towards,gyongyi2004combating,liu2016personalized,Ng:2001:SAL:383952.384003} 
or with respect to specific ranking functions versus specific attacks~\cite{cheng2006manipulability,hopcroft2008manipulation,DBLP:conf/ijcai/NgZJ01}.
Much of the ranking literature has focused on two aspects of fighting spam: link-spam detection~\cite{gyongyi2006link,cheng2006manipulability,DBLP:conf/airweb/AndersenBCHJMT08,DBLP:journals/ton/YuGKX10,yu2008sybilguard,Yu:2011:SDV:2034575.2034593,alvisi2014communities}; and heuristically designing ranking functions that are spam resistant~\cite{kumar2006core, bhattacharjee2006incentive, gyongyi2004combating,krishnan2006web,hopcroft2008manipulation,liu2016personalized}.


We combine the notion of attack detection and attack resistance into a single two-party game that allows us to reason about ranking functions.
\begin{itemize}
\item A spammer can perform a \emph{Sybil attack}\footnote{The moves by the spammer are quite general and subsume specific attacks treated in the literature~\cite{viswanath2011analysis}}:
\begin{itemize}
	\item Create as many nodes in the web graph as they like, for free;
	\item Change the out-links of all nodes they own, for free;
	\item Acquire (by buying, hacking, etc) existing nodes in the web graph, at some cost, except for nodes in a non-empty \emph{trusted set}, which cannot be acquired.
\end{itemize}

\item The ranker can: \begin{itemize}
	\item Modify the ranking function to make it more resistant to spam, subject to quality constraints;
    \item Expend effort to detect if nodes have been acquired by spammers, which we model as raising the cost for spammers to acquire those nodes.
\end{itemize}
\end{itemize}
In this game, the spammer wants to maximize the sum of rank of its nodes, and naturally the ranker wants to minimize this quantity.  The spammer knows the ranking function and the cost function, whereas the ranker does not know what actions the spammer is taking.

Trusted nodes play a critical role.
Without trusted nodes, the ranker cannot hope to resist a spam attack, because the spamming model is general enough that  the spammer could make a complete (or partial) replica of the web  graph, and the ranker would not know which is the correct version.  On the other hand, the ranker cannot simply rank trusted nodes and ignore untrusted nodes because: it is prohibitively expensive to certify that all unspammed nodes are trusted; and, restricting rank to only a few nodes would introduce local (ranking-function) \emph{distortions},  formally defined in Section~\ref{sec:def-distortion}.

This rather broad characterization of spam attacks and defense leaves two obvious questions:
\begin{itemize}
	\item Is there a way to define the spam resistance and distortion of arbitrary ranking functions that is mathematically tractable and matches intuition and practice?
	\item Do these definitions yield mathematical and practical insight?  That is, can we prove that there is a ranking function with both low distortion and high spam resistance? In other words, does the ranker have a good strategy?  
\end{itemize}

\pp{Our contribution.} In this paper, we answer both questions in the affirmative. In Section~\ref{sec:def-spam-resistance},  we refine the object of the spamming game to be the highest cost-benefit ratio the ranker can force on the spammer, which we call the \emph{spam resistance} of a ranking function. In Section~\ref{sec:def-distortion}, we add a further constraint to the game, which is that the ranker must try to minimize \emph{distortion}.  This constraint is meant to avoid trivial ranking functions, such as those that place all rank on a few trusted nodes.  These two sections lean heavily on the practical literature on spam resistance in order to come up with a clean and mathematically interesting definition of the spam game that is mathematically tractable.  

The spammer would seem to have too much power in this game, because it has full knowledge and need not concern itself with the quality of the final ranking.  Indeed, we find that there is an intuitive tradeoff between spam resistance and distortion in that we observe, rather dispiritingly, that all ranking functions in the literature exhibit either high distortion, low spam resistance, or both.


Interestingly, this leaves open the question of the properties of the type of \pr that Google itself used.  To the best of our knowledge, Google has never used uniform reset \pr for web-search ranking, in part because of its obvious susceptibility to spammers.  Instead, Google has used a variant of \pr that we call Min-PPR (and which we define in Section~\ref{sec:results}), because it was intuitively considered to be more spam resistant~\cite{MFC}.  We know of no treatment of Min-PPR in the literature.

In Section~\ref{sec:results}, we set out our main 
contribution: we show that this faith in Min-PPR is well placed.  We prove that Min-PPR exhibits both low distortion and high spam resistance subject to the mild technical condition of fast mixing. Fast mixing is known to hold in real-world social and semantic networks, and if it does not hold then we prove that no variant of \pr has low distortion. 

No ranking function is spam resistant if the ranker makes no effort to detect spam, so we also show that we can compute a cost function for Min-PPR that 
establishes the importance of each node in fighting spam.  In other words, this cost function shows which nodes the ranker should focus its spam-fighting efforts on.

Proving these results requires a substantial technical contribution: while our definitions of distortion and spam resistance are natural, they are also difficult to analyze. The proof that establishes the distortion of Min-PPR relies on a sophisticated analysis of mixing times, 
whereas the proof that establishes the spam resistance of Min-PPR relies on an intricate extremal argument. 
These analyses are left to the Appendices~\ref{sec:distortion} and~\ref{sec:termispam} respectively due to space restrictions.



In summary, in this paper we reimagine the arms race between web rankers and web spammers as a two-person game.  We define the moves of each party.  We show that there is a natural tension between spam resistance and distortion, but we show that the ranker can gain the upper hand; that is, we exhibit a ranking function that has low distortion and high spam resistance.  

\mam{}

\pp{Roadmap.} In Section~\ref{sec:prelim}, we introduce preliminary notation.  In Section~\ref{sec:def-spam-resistance}, we formally define spam resistance and in Section~\ref{sec:def-distortion}, we formally define distortion.  In both sections, we survey existing ranking functions and provide preliminary results.  In Section~\ref{sec:results}, we formally define Min-PPR and give an overview of the main theorems related to its spam resistance and distortion.  In Section~\ref{sec:closure}, we explore the algebraic properties of \pr, which may be of independent interest, and which are used in our main proofs.  
In Section~\ref{sec:related}, we discuss related work. 
In the Appendices~\ref{sec:distortion} and~\ref{sec:termispam}, we prove the theorems outlined in Section~\ref{sec:results}, 
and in Appendix~\ref{sec:experiments}, we present experimental results that validate our theoretical analysis.

%% file: preliminaries.tex
\section{Preliminaries}
\label{sec:prelim}


Let $\graph{G}=(\nodes,\edges)$ be an $n$-node graph. Since we will frequently discuss random walks on graphs, we always assume that any node with no outgoing edges to other nodes has a self-loop. For any edge $(u,v) \in \edges$, $v$ is an \emph{out-neighbor} of
$u$, and $u$ is an \emph{in-neighbor} of $v$. 
For any $v \in \nodes$, let $\IN(v)$ and $\OUT(v)$ denote the set of in-neighbors and out-neighbors of $v$, respectively.
Note that $|\OUT(v)|>0$ for all $v$, because of the self loops. Let $\din(v) = |\IN(v)|$, and $\dout(v) = |\OUT(v)|$. 

We extend the $\min$ operator on vectors to component-wise min, so that, for vectors $\scaledvect{x_1}, \ldots, \scaledvect{x_k}$, 
\[
\min\{\scaledvect{x_1}, \ldots, \scaledvect{x_k}\}[i] = \min\{\scaledvect{x_1}[i], \ldots, \scaledvect{x_k}[i]\}.
\]
For any real vector $\scaledvect{x}$, let $||\scaledvect{x}||:=\sum_i{|\scaledvect{x}_i|}$ be its $\ell_1$-norm, and if $\scaledvect{x} \ne \scaledvect{0}$ then we define $\norm{\scaledvect{x}} := \scaledvect{x}/||\scaledvect{x}||$ to be the standard $\ell_1$-normalization. Let the \emph{support} of $\scaledvect{x}$ be the set of coordinates $i$ with $\scaledvect{x}[i] \ne 0$. 

We define a \textit{ranking vector} for a graph $\graph{G}=(\nodes,\edges)$ to be any function $\vect{x}\colon \nodes \to [0,1]$ such that $\sum_{v \in \nodes} \vect{x}[v] = 1$. (We normalise our rankings to facilitate easy comparison.) Note that any ranking vector also defines a probability distribution on $\nodes$. We also adopt the convention that for all $\mathbf{A} \subseteq \nodes$, $\vect{x}[A] = \sum_{v \in \mathbf{A}}\vect{x}[v]$. 

Let $\ergodicG$ be the set of all directed graphs on which the uniform random walk is ergodic. For all $\graph{G} \in \ergodicG$, we write $\station{\graph{G}}$ for the stationary distribution of the uniform random walk; we call this the \emph{reference rank} of $\graph{G}$. 

We define \pr on an arbitrary graph $\graph{G} = (\nodes,\edges)$ as follows. Let $\rp \in (0,1)$, and let $\R$ be a probability distribution on $\nodes$. 
We write $\R[v]$ for the probability of choosing $v \in \nodes$ under $\R$. The \pr random walk's state space is $\nodes$. From any node $v \in \nodes$, with probability $1-\rp$ it traverses a uniformly random out-edge from $v$, and with probability $\rp$ it moves to a node drawn from $\R$; this event is called a \emph{reset}\footnote{Also known in the literature as \emph{teleportation.}}. We denote this random walk model by $(\graph{G},\R,\rp)$. The stationary distribution of $(\graph{G},\R,\rp)$ 
is the \emph{\pr} of $\graph{G}$ with \emph{reset probability} $\rp$ and \emph{reset vector} $\R$. We denote the value of this PageRank at a node $v \in \nodes$ by $\station{\graph{G},\R,\rp}[v]$; note that $\station{\graph{G},\R,\rp}$ is a ranking vector. (Since $(\graph{G},\R,\rp)$ is ergodic for all choices of parameters, $\station{\graph{G},\R,\rp}[v]$ is uniquely defined.) 

The original version of \pr defined by Brin and Page~\cite{page1999pagerank} takes $\R$ to be the uniform distribution, so that $\R[v] = 1/|\nodes|$ for all $v \in \nodes$. We denote this vector by $\uniform$, and call this version of \pr the \emph{uniform \pr (\UPR)}. \emph{Personalized \pr (PPR)} instead takes $\R$ to be of the form $\R[c] = 1$ for some $c \in \nodes$, which is called the \emph{center node}. Thus each time the walk $(\graph{G},\R,\rp)$ resets, it returns deterministically to $c$. We denote the resulting ranking vector by $\station{\graph{G},c,\rp}$. 

We define a \emph{ranking algorithm} to be any algorithm that takes as input an arbitrary graph $\graph{H}$ and an arbitrary non-empty \emph{trusted set} $\trust \subseteq \nodes_\graph{H}$ and outputs a ranking vector. 
 For example,
$\mTPPR$ is
the version of \PPR in which the center vertex is trusted.

We will see in Observation~\ref{obs:trust-needed} that any ranking algorithm that does not make use of $\trust$ has zero spam resistance. 

For all graphs $\graph{G} \in \ergodicG$ and all probability distributions $\vect{p}$ and $\vect{q}$ over $\nodes_\graph{G}$, we define the \emph{total variation distance} between $\vect{p}$ and $\vect{q}$ by
\begin{align*}
    \dtv(\vect{p},\vect{q}) &:= \frac{1}{2}\sum_{v \in \nodes_\graph{G}}\big|\vect{p}[v] - \vect{q}[v]\big|.
\end{align*}
In the following definition, let $\vect{p}_{t,v}$ be the distribution of the uniform random walk on $\graph{G}$ from initial state $v \in \nodes_\graph{G}$ at time $t \ge 0$. Then for all $\rho>0$, the (worst-case) \emph{mixing time} of $\graph{G}$ to within error $\rho$ is given by
\begin{align*}
    \mix_\graph{G}(\rho) &:= \min\Big\{t \ge 0 \colon \mbox{for all }v \in \nodes_\graph{G},\ \dtv\big(\vect{p}_{t,v},\, \station{\graph{G}}\big) \le \rho\Big\}.
\end{align*}
Following standard practice, we 
take $\rho$ to be $1/4$ when we don't specify it.

%% file: def-spam-resistance.tex
\section{Defining Spam Resistance}\label{sec:def-spam-resistance}

We first give some examples of ranking algorithms that show  differing levels of spam resistance. We then define spam resistance formally, we note that it matches our intuition on the examples provided, and we prove some preliminary results. 

\UPR would appear to be trivial to spam. Consider, for example, the effect of making a large set of new nodes with self-loops. The more new nodes the spammer makes (for free), the larger the probability that the \pr random walk will reset to those nodes, and the more rank the spammer will capture. In fact, this sort of attack is known to be viable in practice~\cite{DBLP:conf/airweb/Baeza-YatesCL05,DBLP:conf/vldb/GyongyiG05}.

In 1999, Kleinberg~\cite{kleinberg1999hubs}
introduced the notion of \emph{hub} and \emph{authority scores} and
the HITS algorithm to compute them.  Intuitively, a page receives a
high authority score if it is pointed to by many high-quality hubs.  A
page receives a high hub score if it points to many high-scoring
authorities.  Hub scores depend on out-links
and are therefore free to spam by creating new nodes and pointing them to nodes with high authority scores.    But once a spammer owns many pages with
high hub scores, it can create pages with high authority scores, once
again for free.  Such considerations are far from hypothetical.  This vulnerability was
already well understood in 2004~\cite{ilprints646}, and 
Assano et.\ al~\cite{asano2007improvements} report that HITS was
unusable by 2007, due to its spammability. Like UPR, we should expect HITS to have no spam resistance, as the spammer can acquire high rank at no cost.

Next, consider \emph{SuperTrust}, which assigns non-zero rank only to trusted nodes, which are known not to belong to spammers. SuperTrust is unspammable, because the spammer receives no rank no matter what they do. (Of course, because in general very few nodes can be fully trusted, in most cases SuperTrust has very high distortion and is not suitable as a ranking algorithm.)





Finally, consider 
$\mTPPR$, 
the version of \PPR in which the center vertex is trusted and cannot be subverted by a spammer. It is not difficult for the spammer to arrange things so that whenever the \pr random walk enters a vertex they own, it remains in spammer-owned territory until the next reset; however, the random walk will only enter spammer-owned vertices $v$ at a rate commensurate with $v$'s true \PPR. Thus in order to acquire significant rank, the spammer has to invest time and effort into acquiring nodes which already have significant rank, and we should expect $\mTPPR$ to be quite strongly spam resistant.

We now give a definition of spam resistance that coincides
with our intuition that HITS and \UPR are not at all spam resistant, $\mTPPR$ is quite strongly spam resistant, and SuperTrust has unbounded spam resistance.

\pp{A formal definition of spam resistance.}
In the examples given above, there was no mention of how much it costs to acquire a node.  Here, we make explicit the cost model,  the changes a spammer can make to a graph, and the cost/benefit ratio a spammer can achieve.  It is this cost/benefit ratio that defines spam resistance.

Let $\graph{G} = (\nodes_\graph{G},\edges_\graph{G})$, and let $\trust_\graph{G}\subseteq\nodes_\graph{G}$ be the set of trusted sites.  Then $\nodes_\graph{G}\setminus\trust_\graph{G}$ is the set of all sites that might be acquired by the spammer. 
For all $\purchase \subseteq \nodes_\graph{G}\setminus\trust_{\graph{G}}$, let $\Pgraphs$ be the set of all graphs obtainable from $\graph{G}$ by:
\begin{itemize}
  \item Adding an arbitrary (possibly empty) set $\spam$ of new vertices;
  \item Changing all the out-edges of vertices in $\spam\cup\purchase$ in an arbitrary fashion.
\end{itemize}
Thus $\Pgraphs$ is the set of all graphs that the spammer can obtain after acquiring the vertices in $\purchase$. For any $\graph{H}=(\nodes_\graph{H},\edges_\graph{H})\in\graph{G}_{\purchase}$, let $\spam_\graph{H}$, the set of \emph{spam nodes of $\graph{H}$}, be $\nodes_\graph{H} \setminus \nodes_\graph{G}.$ 

We are now ready to define the cost and benefit of spamming. Like ranks, we will normalize costs in order to make them comparable.

Let $f$ be a ranking algorithm, so that $f(\graph{H},\trust)$ returns a rank vector for all graphs $\graph{H}$ and all non-empty $\trust \subseteq \nodes_\graph{H}$.  
We think of $\graph{H}$ as a post-spam graph, belonging to some $\graph{G}_\purchase$; thus the ranking algorithm does not know which nodes are owned by the spammer but does know which nodes are trusted. 
We extend $f(\graph{H},\trust)$ from vertices to sets in the standard way: 
For all $\mathbf{X} \subseteq \nodes_\graph{H}$, $f(\graph{H},\trust)[\mathbf{X}] = \sum_{v\in\mathbf{X}}f(\graph{H},\trust)[v]$.

A \emph{cost function} is a normalized function on $\nodes_\graph{G}\setminus\trust_\graph{G}$, i.e.\ 
$C(v) \ge 0$ for all $v \in \nodes_\graph{G}\setminus \trust_\graph{G}$ and $\sum_{v \in {\nodes_\graph{G}\setminus \trust_\graph{G}}} C(v) = 1$.    
Let $\set{C}_{(\graph{G},\trust_\graph{G})}$ denote the set of all such cost
functions.%
\footnote{Note that we restrict the spammer to acquiring whole nodes, rather than individual edges. This is because $\spam$ could be arbitrarily large, so the set of possible edges between $\nodes_\graph{G}\setminus \trust_\graph{G}$ and $\spam$ is also arbitrarily large, and under an edge-acquisition model we would be unable to normalize costs. For every ranking algorithm considered in this paper, this technical restriction makes no difference to spammability.}  
We extend this definition to subsets of nodes as above. 

We are ready to define spam resistance, as follows:

\begin{definition}\label{def:spamres} For all classes $\mathcal{A}$ of graphs:
\begin{enumerate}
\item For all $\sr > 0$, a ranking algorithm $f$ is \emph{$\sr$-spam resistant on $\mathcal{A}$} if, for all $\graph{G} = (\nodes_\graph{G},\edges_\graph{G}) \in \mathcal{A}$ and all non-empty $\trust_\graph{G} \subseteq \nodes_\graph{G}$, there exists a cost function $C \in \set{C}_{(\graph{G},\trust_\graph{G})}$ such that, for all $\purchase\subseteq \nodes_{\graph{G}}\setminus \trust_{\graph{G}}$ and $\graph{H}\in \graph{G}_\purchase$ with $f(\graph{H},\trust_\graph{G})[\spam_\graph{H}\cup\purchase] > 0$, 
\[\frac{C(\purchase)}{f(\graph{H},\trust_\graph{G})[\spam_\graph{H}\cup\purchase]} \geq \sr. \]

\item 

A ranking algorithm $f$ has \emph{unbounded spam resistance on $\mathcal{A}$} if, for all $\graph{G} = (\nodes_\graph{G},\edges_\graph{G}) \in \mathcal{A}$, all non-empty $\trust_\graph{G} \subseteq \nodes_\graph{G}$, all $\purchase \subseteq \nodes_\graph{G} \setminus \trust_\graph{G}$, and all $\graph{H} \in \graph{G}_\purchase$,
\[f(\graph{H},\trust_\graph{G})[\spam_\graph{H}\cup\purchase] = 0.\]

\item A ranking algorithm $f$ has \emph{zero spam resistance on $\mathcal{A}$} if, for all $\graph{G} = (\nodes_\graph{G},\edges_\graph{G}) \in \mathcal{A}$, there exists a non-empty $\trust_\graph{G} \subseteq \nodes_\graph{G}$ and $\graph{H} \in \graph{G}_\emptyset$ such that $f(\graph{H},\trust_\graph{G})[\spam_\graph{H}] > 0$.

\end{enumerate}
\end{definition}

\pp{Intuition of the cost function.} The cost function captures how difficult it is for a spammer to subvert a node in a way that is hidden from the spam-detection efforts of the ranker.  
Thus, for example, a search engine would be able to substantially increase the cost of a specific node by assigning a human to watch it carefully and to de-index it if they suspected it had been acquired by a spammer, and in fact search engines do manipulate their interaction with spammers by such methods~\cite{GoogleWebSpam}. The correct way to view the cost function is as a guarantee that the ranker has a strategy to force a spammer to pay for their rank; using a spam-resistant ranking function by itself is not enough.  Thus, when we establish that a ranking function is spam resistant, we must exhibit a cost function that the ranker can efficiently approximate, then use in conjunction with that ranking function.


Indeed, while the definition of spam resistance only requires that some good cost function exist, all our proofs of spam resistance will construct explicit cost functions to direct the ranker's spam-detection effort. Observe that no amount of such effort is sufficient if the ranking function has zero spam resistance and that no such effort is needed if the ranking function has unbounded spam resistance. An examination of SuperTrust and UPR supports this interpretation.



\pp{Preliminary results.} As a warm-up, to see that trusted sites are necessary, consider any ranking algorithm $f$ whose output does not depend on $\trust$. Let $\graph{G}$ be an arbitrary graph, and let $\graph{H}$ consist of two disjoint copies of $\graph{G}$ spanning vertex sets $\nodes_1$ and $\nodes_2$. Then given $\graph{H}$ as input, without loss of generality   $f$ assigns rank at least $1/2$ to $\nodes_1$. Viewing $\graph{H}$ as a spam graph in $\graph{H}[\nodes_2]_\emptyset$, where the base graph $\graph{H}[\nodes_2]$ is the copy of $\graph{G}$ spanning $\nodes_2$, $\purchase = \emptyset$, and 
$\spam_\graph{H} = \nodes_1$,
we see that the spammer has attained rank at least $1/2$ without buying any vertices. Thus we have proved the following.

\begin{observation}\label{obs:trust-needed}
    Any ranking algorithm  that is invariant under membership in $\trust_\graph{G}$ has zero spam resistance on all graph classes.
\end{observation}

Observation~\ref{obs:trust-needed} implies that UPR and HITS have zero spam resistance. On the positive side, any ranking algorithm that assigns non-zero rank only to vertices in $\trust_\graph{G}$ (that is, any variant of SuperTrust) has unbounded spam resistance; in fact, these are the only ranking algorithms with unbounded spam resistance. Moreover, the trusted variant $\mTPPR$ of PPR also has high spam resistance, as follows.

\newcommand{\statepprsimple}
{
    For all $\rp \in (0,1)$, $\mTPPR$ is $\rp$-spam resistant on all graph classes.  The cost function that establishes this spam resistance is the $\mTPPR$ itself, normalized over untrusted nodes.
}
\begin{lemma}\label{lem:ppr-simple}
    \statepprsimple
\end{lemma}

We prove this in Section~\ref{sec:termispam}, but for now, suppose the spammer acquires a node $v$ with $\mTPPR(v) = x$
and redirects its outward edges into the set $\spam \cup \purchase$ of vertices they own. Then each time the random walk associated with $\mTPPR$ hits $v$, which happens at rate roughly $x$, at worst it will stay in $\spam \cup \purchase$ until the walk next resets, which takes $1/\rp$ time on average. So the spammer should acquire rank at most $x/\rp$. 
As we might expect, Lemma~\ref{lem:ppr-simple} is essentially tight.


\begin{observation}\label{obs:Thm3istight}
    For all $\rp,\eta \in (0,1)$, $\mTPPR$ is not $\rp\cdot (1+\eta)/(1-\rp)$-spam resistant on the class of all~cliques.
\end{observation}

%% file: def-distortion.tex
\section{Defining Distortion}\label{sec:def-distortion}

In this section, we define the \emph{distortion} of a ranking vector. We first discuss some guiding principles, then give a formal definition, then conclude with a discussion of the distortion of \PPR and \UPR.


\pp{Guiding principles.} 
First, we would like to give an accurate ranking for every site, so our metric should be mostly concerned with the maximum error at any site rather than the total error across all sites. Note in particular that total variation distance from $\station{\graph{G}}$ does not capture this well --- a total variation distance of $0.1$ from $\station{\graph{G}}$ could indicate anything from an additive error of $0.1/|\nodes_\graph{G}|$ at every vertex (a relatively minor error) to an additive error of $0.1$ at a single vertex (a very severe error).

Second, we should be concerned with multiplicative error rather than additive error. To see this, suppose $\graph{G} = (\nodes,\edges)\in \calG$ has $n$ vertices, let $v_1, v_2 \in \nodes$, and suppose that $v_1$ has \reference rank $1/\sqrt{n}$ and $v_2$ has \reference rank $1/4$. Then intuitively, assigning $v_1$ a rank of $1/\log n$ is a far more severe mistake than assigning $v_2$ a rank of $1/4 - 2/\log n$, even though the additive error is smaller. Likewise, assigning $v_2$ a rank of $1/\sqrt{n}$ would be a far more severe mistake than assigning it a rank of $1/8$, even though the additive errors are comparable. 

Third, multiplicative error is only significant when it causes us to make important mistakes in the final site ranking. To illustrate what we mean by this, suppose our ranking function were to assign rank $1/2^{n/2}$ to a node $v$ whose \reference rank is $1/2^{n}$. This would constitute a huge multiplicative error of $2^{n/2}$.
However, in practice, $1/2^n$ and $1/2^{n/2}$ are both so small as to be indistinguishable, so this error is unlikely to have much of an impact on query rankings. In general, we can safely ignore multiplicative error on vertices with ``insignificant'' reference rank as long as we still assign them ``insignificant'' rank.

\pp{Formal definition.} Let $\dist > 0$. For all $n$-vertex graphs $\graph{G} = (\nodes,\edges) \in \ergodicG$, all ranking vectors $\vect{x}\colon \nodes\to [0,1]$, and all vertices $v \in \nodes$, define the \emph{stretch} and \emph{contraction of $\vect{x}$ on $v$} by
\begin{align*}
    \str_\dist(\vect{x},\graph{G},v) =
        \frac{\max\{\vect{x}[v],1/n^\dist\}}{\max\{\station{\graph{G}}[v],\,1/n^\dist\}}
&&    
    \cont_\dist(\vect{x},\graph{G},v) = 
        \frac{\max\{\station{\graph{G}}[v],1/n^\dist\}}{\max\{\vect{x}[v],\,1/n^\dist\}}.
\end{align*}
We then define\footnote{The names of our accuracy metrics are taken from the theory of metric space embedding. In this theory, distortion is stretch times contraction, but since we consider normalized ranking functions, we instead take the distortion to be the maximum of the stretch and contraction.} the \emph{distortion of $\vect{x}$ on $v$ } 
by
\[
    \distortion_\dist(\vect{x},\graph{G},v) = \max\{\str_\dist(\vect{x},\graph{G},v),\ \cont_\dist(\vect{x},\graph{G},v)\},
\]
and the \emph{distortion of $\vect{x}$ on $\graph{G}$} by
\[
    \distortion_\dist(\vect{x},\graph{G}) = \max\{\distortion_\dist(\vect{x},\graph{G},v) \colon v \in \nodes\}.
\]

We pause for a moment to map these definitions back onto our intuition. We take our \reference rank threshold for a vertex to be ``significant'' to be $1/n^\dist$. Observe that if a vertex $v$ is significant, and $\vect{x}$ ranks it as significant, then the distortion of $\vect{x}$ on $v$ is simply the approximation ratio between its $\vect{x}$-rank and its \reference rank. The maxima in the definitions of $\str$ and $\cont$ capture the idea that we should disregard any multiplicative error below the significance threshold. For example, if a vertex is insignificant, and $\vect{x}$ ranks it as insignificant, then the distortion of $\vect{x}$ on $v$ is 1, i.e.\ as low as possible. 

It is natural to ask which vertices should we consider insignificant. That is, how should we choose $\dist$? It will turn out that for our purposes, it doesn't actually matter --- our main distortion bounds hold for any choice of $\dist$, so we leave this question to future work. Note, however, that since the total reference rank is 1, any choice of $\dist<1$ will leave at most an $n^{\dist-1} = o(1)$ proportion of nodes above the significance threshold; for this reason we shall take $\dist \ge 1$. 

\pp{Restricting to fast-mixing graphs.} Before we discuss specific examples, we  note the following intuitive requirement for the output of any PageRank to have low distortion on its input graph $\graph{G} \in \ergodicG$: After a reset, on average, the distribution of its random walk should have time to converge to $\station{\graph{G}}$ before the next reset. Thus for a reset vector $\R$ and a reset probability $\rp$ to yield an effective PageRank on $\graph{G}$, the mixing time of the uniform random walk on $\graph{G}$ whose initial state is drawn from $\R$ should be less than $1/\rp$. As an example of what might otherwise go wrong, consider the case where $\graph{G}$ is an $n$-vertex directed cycle --- where unless $\R$ is close to uniform, or $\rp$ is very small, the resulting PageRank will be biased away from segments with low mass in $\R$. 

Fortunately, real networks are very often fast mixing (at least on the giant component). Experimental studies~\cite{albert1999internet,broder2000graph,faloutsos1999power} have demonstrated that the degrees of the web graph are power-law distributed; Gkantsidis, Mihail and Saberi~\cite{GMS03} prove that $n$-vertex random power-law graphs have $O(\log n)$ mixing time~\cite{GMS03} with high probability. While power-law random graphs are of foundational importance, there are many other models for random ``web-like'' graphs, including graphs such as the Facebook graph, which may not admit a power-law degree distribution~\cite{gjoka2010walking,ugander2011anatomy}. Among the most well-known such models are preferential attachment~\cite{barabasi1999emergence}, which exhibits $O(\log n)$ mixing time with high probability~\cite{COOPER2007269,mihail2006certain}, and the Newman--Watts small world model~\cite{newman1999renormalization}, which exhibits $O(\log^2 n)$ mixing time with high probability~\cite{addario2012mixing}. Fast mixing is also a common assumption in the literature on defenses against Sybil attacks~\cite{mohaisen2010measuring,viswanath2011analysis}. Some important models, such as random hyperbolic graphs~\cite{hyperbolicgraphs}, do not exhibit worst-case fast mixing~\cite{KM2018}, though it is not known if they have fast average-case mixing time. 

Mohaisen, Yun and Kim~\cite{mohaisen2010measuring} offer an explanation for why some models exhibit fast worst-case mixing and others do not, by characterizing two kinds of social networks: those in which nodes are linked based on real acquaintance, such as DBLP, and those that do not have this requirement, such as Facebook. They argue based on experimental evidence that DBLP-like networks are slowly mixing compared to the Facebook-like networks in the worst case, but are nevertheless still fast-mixing in the average case. 

%
As we now show, PPR and UPR do not in general output low-distortion ranking vectors even on very fast-mixing graphs, but a mixing time a little lower than $1/\rp$ will suffice for Min-PPR (see Section~\ref{sec:results}). For any function $T\colon \mathbb{N}\to [0,\infty)$, let
\[
    \mixclass{T} := \big\{\graph{G} \in \ergodicG \colon \mix_\graph{G}(1/4) \le T(|\nodes_\graph{G}|)\big\}.
\]

\pp{PPR usually has high distortion.} Consider 
\PPR with center $c \in \nodes$ and reset probability $\rp$. Since the \PPR random walk resets to $c$ with probability $\rp$ at each step, we have $\station{\graph{G},c,\rp}[c] \ge \rp$, so 
\[
	\distortion_\dist(\station{\graph{G},c,\rp},\graph{G}) \ge \str_\dist(\station{\graph{G},c,\rp},\graph{G},c) \ge \rp/\max\{\station{\graph{G}}[c],1/|\nodes_\graph{G}|^\dist\}.
\]
Thus PPR has high distortion unless either $\rp$ is unrealistically small or $c$ happens to be a vertex with extremely high reference rank. Indeed, since the total reference rank of $\graph{G}$ is~1, we obtain the following.


\begin{observation}\label{obs:PPR-bad}
	Let $\dist \ge 1$, $\rp \in (0,1)$, and $t \le \epsilon n^\dist$. 
	Then for any $n$-vertex graph $\graph{G} \in \ergodicG$, there are at most $t/\epsilon$ vertices $c \in \nodes_\graph{G}$ such that $\distortion_\dist(\station{\graph{G},c,\rp}, \graph{G}) \le t$.
\end{observation}


Moreover, a clique is a simple example of an $n$-vertex graph on which the output of \PPR has distortion at least $\rp n$ for any choice of center. (Note that all cliques are contained in $\mixclass{1}$.) While there do exist specific graphs and center choices for which the output of \PPR has low distortion, we have no reason to believe that these specific inputs are relevant for real-world use.





\pp{UPR can have high distortion.} For any $\graph{G} \in \ergodicG$, $v \in \nodes$, and $(u,v) \in \edges$, replacing $(u,v)$ with a suitably large collection of internally vertex-disjoint two-edge paths from $u$ to $v$ inflates $\station{\graph{G},\uniform,\rp}[v]$ to a near-arbitrary extent while leaving $\station{\graph{G}}[v]$ almost unchanged. This construction does not significantly affect the mixing time of $\graph{G}$, so we conclude the following.

\begin{observation}\label{obs:UPR-bad}
	Let $\dist \ge 1$ and $\rp \in (0,1)$. Then there exist infinitely many graphs $\graph{G} =(\nodes,\edges)\in \ergodicG_4$ with $\distortion_\dist(\station{\graph{G},\uniform,\rp},\graph{G}) \ge \tfrac{1}{2}\rp(1-\rp)|\nodes_\graph{G}|^\dist$.
\end{observation}

%% file: results.tex
\section{Min-PPR has low distortion and high spam resistance}
\label{sec:results}

In this section, we introduce \minppr, which is the idea that 
ranks should be determined by running PPR multiple times, with different trusted centers.  Recall that this is the actual version of \pr used in Google.
The rank of a site~$w$ should be the minimum of the ranks that result (appropriately normalised).
We will  show how to turn the
\minppr idea into a family of ranking algorithms,  $\mtermialgot{k}{\rp}$, as Min-PPR on $k$ trusted centers with reset probability $\epsilon$, though we will need some constraints on how the centers are selected.  We then outline the theorems needed to establish that it has low distortion and high spam resistance. These theorems are proven in Appendices~\ref{sec:distortion} and~\ref{sec:termispam}.    

Specifically, we will show that, given a small enough $\epsilon$ and a large enough $k$, $\mtermialgot{k}{\rp}$ is $(\rp/3k)$-spam resistant.  
It turns out (see the discussion following the statement of Theorem~\ref{thm:min-spam-resistant})
that  T-PPR$_\rp$ is $\rp$-spam resistant but not $(2\rp)$-spam resistant. Thus, we see that $\mtermialgot{k}{\rp}$ resists spam almost as well as $\mbox{T-PPR}_\rp$.  Moreover, we prove that there are many possible choices of $\trust_\graph{G}$ such that $\mtermialgot{k}{\rp}(\graph{G},\trust_\graph{G})$ has $1+o(1)$ distortion, so that $\mtermialgot{k}{\rp}$ is accurate on the pre-spam graph. 
Finally, we will show that  $\mtermialgot{k}{\rp}$ is a \pr, by the closure properties of \pr that we establish below.  Hence, it fits neatly into the existing methods and heuristics of the field. 

\pp{Defining \tminppr.} In order to compute  $\mtermi{\graph{G},\centers,\rp}$, 
we choose  an arbitrary subset $\centers$ of $\trust$  
of size $\min\{k,|\trust|\}$ 
which does not depend on $\graph{G}$ and output
$$
	\mtermi{\graph{G},\centers,\rp}[w] = \norm{\min\{\station{\graph{G},c,\rp}[w] \colon c \in \centers\}}.
$$
There is a technical difficulty here: if $\graph{G}$ has been disconnected by the spammer and the centers~$\centers$ are badly chosen, it might be that for all vertices $w$ we have $\min\{\station{\graph{G},c,\rp}[w] \colon c \in \centers\} = 0$, in which case the normalisation above is invalid and $\mtermi{\graph{G},\centers,\rp}$ does not output a ranking vector.
To deal with this, we introduce the following technical restriction.

\begin{definition}\label{def:conherence}
	Let $\graph{G} = (\nodes, \edges)$ be a (directed)  graph and let $\centers \subseteq \nodes$. We say that $\centers$ is \emph{coherent} if it is non-empty and, for some $w \in \nodes$, there is a path in $\graph{G}$ from each vertex in $\centers$ to $w$. 
\end{definition}

If $\graph{G} \in \calG$ then every non-empty subset of $\nodes$ is coherent, but to turn Definition~\ref{def:conherence} into a family of ranking algorithms we must consider what happens if the spammer breaks this coherence. For each positive integer $k$ and each real number $ \eps \in (0,1)$, we now define the ranking algorithm 
$\mtermialgot{k}{\rp}$  as follows. (As in  $\mTPPR$, the “T” stands for “trusted”.)

Let $\graph{G} = (\nodes,\edges)$ be a graph and let $\trust \subseteq \nodes$ be a non-empty set of trusted nodes. 
The algorithm chooses
  an arbitrary subset $\centers$ of $\trust$  
of size $\min\{k,|\trust|\}$ which does not depend on $\graph{G}$. It then chooses an arbitrary maximum-size set $\centers' \subseteq \centers$ which is coherent in $\graph{G}$, and outputs $\stationmin{\graph{G},\centers',\rp}$. (Note that some choice of $\centers'$ must exist, since any 1-vertex set is coherent, and that we have $\centers' = \centers$ unless the spammer has disrupted the coherence of $\centers$.)

 Perhaps surprisingly, the output   of 
$\mtermialgot{k}{\rp}$ 
 is not just a normalized minimum of PageRanks, but a PageRank in itself.




\begin{lemma}\label{lem:min-is-PR}
    Let $\graph{G} = (\nodes,\edges)$ be a graph, let $\rp \in (0,1)$, and let $\centers \subseteq \nodes$ be coherent. Then there exists a reset vector $\R$ such that $\mtermi{\graph{G},\centers,\rp} = \station{\graph{G},\R,\rp}$. 
\end{lemma}

In fact, we prove a stronger result as Theorem~\ref{thm:closedmin}: the class of \prs is closed under normalized component-wise min whenever the component-wise min is not identically zero.



\pp{\minppr has low distortion.} Our first result says that on suitably fast-mixing graphs, \emph{any} \pr has low contraction.

\newcommand{\statecontrlow}{
    Let $\dist > 0$, let $\rp \in (0,1)$, and let $T(n)$ be any function 
such that, for all $n$,  
$$0 \le T(n) \le 1/(2\rp(3+\dist\log_2 n)).$$ 
Then for all $n$-vertex graphs $\graph{G} \in \mixclass{T}$, all reset vectors $\R$ on $\graph{G}$, and all $y \in \nodes_\graph{G}$, we have $\cont_\dist(\station{\graph{G},\R,\rp},\graph{G},y) \le 1 + 2\rp T(n)(3+\dist\log_2 n)$.
}
\begin{lemma}\label{lem:contr-low}
    \statecontrlow
\end{lemma}


Lemma~\ref{lem:contr-low} confirms our intuition about the behavior of PPR: On fast-mixing graphs, its inaccuracy is solely the result of a large spike of bias around its trusted center, which \minppr can correct for. We will use this result in our analysis of \minppr. As a corollary (see Theorem~\ref{thm:dtv-low}), we see that any \pr has total variation distance at most $\rp T(n)(3+\log_2 n)$ from $\station{\graph{G}}$. 

We cannot hope for $\mtermi{\graph{G},\centers,\rp}$ to have low distortion for an arbitrary (coherent) choice of $\centers$, since, if the vertices of $\centers$ are clustered together, then their distortion spikes may overlap and cause \minppr to suffer the same distortion as PPR. But it is nevertheless true that good choices of $\centers$ are very common and easy to find. 

\newcommand{\stateminpprdist}{
    Let $\dist \ge 1$. Let $\rp \in (0,1)$ and let $T(n) \le 1/(210\rp\dist\log_2 n)$. Let $\graph{G} \in \ergodicG$ be an $n$-vertex graph with $n \ge 3$, and suppose that the worst-case mixing time of $G$ is at most $T(n)$. Let $k\ge 1$, let $\R$ be an arbitrary reset vector, let $X_1, \dots, X_k$ be drawn independently from $\nodes_\graph{G}$ with probabilities given by $\station{\graph{G},\R,\rp}$, and let $\centers = \{X_1, \dots, X_k\}$. Then with probability at least $1-4^{-k}n$, the distortion of $\mtermialgot{k}{\rp}(\graph{G},\centers)$ satisfies
    $\distortion_\dist(\mtermi{\graph{G},\centers,\rp},\graph{G}) \le 1 + 210\rp\dist T(n)\log_2 n$.
}
\begin{theorem}[\textbf{Main Result}]\label{thm:minppr-dist}
    \stateminpprdist
\end{theorem}


Thus according to $\station{\graph{G}}$, \UPR, or any other PageRank, $\Theta(\log n)$-sized sets of centers giving rise to accurate Min-PPR's are very common. (The reason we can afford to be so flexible in our choice of distribution is that all \prs are close in total variation distance, as stated above.) 
Note that the dependence of Theorem~\ref{thm:minppr-dist} on our significance parameter $\dist$ is very mild. 

Recall that Theorem~\ref{thm:minppr-dist} requires fast worst-case mixing, i.e.\ that the uniform random walk mixes quickly from every vertex in $\graph{G}$. While this is a common assumption, as discussed in Section~\ref{sec:def-distortion}, some web-like graphs may exhibit only fast average-case mixing. For this reason, as Theorem~\ref{thm:minppr-dist-hard}, we prove a version of Theorem~\ref{thm:minppr-dist} which requires only fast average-case mixing from our chosen centers. We prove the result by altering $\mtermialgot{k}{\rp}$ to use only a carefully-chosen subset of the trusted vertices.

\pp{\tminppr has high spam resistance.} As noted above, \minppr has zero spam resistance by Observation~\ref{obs:trust-needed}. \tminppr, however, is highly spam resistant.


\newcommand{\stateminspamresistant}{
    For any $\rp \in (0,1)$ and any positive integer $k$,  $\mtermialgot{k}{\rp}$ is $(\rp/3k)$-spam resistant on $n$-vertex graphs in $\ergodicG$ with worst-case mixing time at most $1/(3\rp(3+\log_2n))$. A cost function that establishes this spam resistance is the average of the cost functions of the component T-PPRs.
}
\begin{theorem}[\textbf{Main Result}]\label{thm:min-spam-resistant}
    \stateminspamresistant
\end{theorem}

Recall from Lemma~\ref{lem:ppr-simple} that T-PPR$_\rp$ is at least $\rp$-spam resistant on all graph classes, and recall from Observation~\ref{obs:Thm3istight} that this is close to tight even on fast-mixing graphs --- for example, T-PPR$_\rp$ is not $(1.01\rp/(1-\rp))$-spam resistant even on the class of cliques, which is contained in $\mixclass{1}$.
Since Theorem~\ref{thm:min-spam-resistant} shows that $\mtermialgot{k}{\rp}$ is $(\rp/3k)$-spam resistant on suitably fast-mixing graphs, we conclude that in this setting $\mtermialgot{k}{\rp}$ inherits most of the spam resistance of~T-PPR$_\rp$. 



\textbf{Min versus Median.} It is natural to ask whether, instead of taking the normalized component-wise minimum of our PPRs, we could take the normalized component-wise median. We show that Min-PPR has a crucial advantage over Median-PPR: The normalized minimum of any set of PageRanks with reset probability $\rp$ is itself a PageRank with reset probability $\rp$, whereas the normalized median of any set of \prs with reset probability $\rp$ is a \pr with possibly much larger $\rp$.  

This preservation of $\rp$ is important, as without it the closure condition 
that we introduce in the next section
to show that  $\mtermialgot{k}{\rp}$ is a \pr
would be so weak as to be useless. Notice that, as the reset probability is allowed to become arbitrarily large, the resulting PageRank will approach the reset vector, with little contribution from the underlying graph --- indeed, we show that any strictly positive vector whose entries sum to $1$ is a PageRank.

In Section~\ref{sec:closure}, we make a distinction between operators that are \emph{strongly closed}, which means they preserve the reset probability, and those that are \emph{weakly closed}, which means that they might not. We give simple necessary and sufficient conditions for a vector to be a \pr with a given reset probability, and show that min is strongly closed for \pr (Theorem~\ref{thm:closedmin}), whereas median is only weakly closed (Lemma~\ref{lem:median-bad} and~\ref{lem:median-weak-closure}). Indeed, the reset probability of Median-PPR may be as high as 1/2 even when $\rp$ is arbitrarily small.


\pp{Summary.} Overall, suppose that our $n$-vertex pre-spam graph $\graph{G}$ lies in $\ergodicG_{\textnormal{polylog}(n)}$. Choose our significance threshold $1/n^\dist$ arbitrarily subject to $\dist \ge 1$, and take $k = \Theta(\log n)$ and $\rp = 1/\mbox{polylog}(n)$. (This requires a rough estimate of $n$, but this should not be a major obstacle in practice.) Then Theorem~\ref{thm:min-spam-resistant} implies that $\mtermialgot{k}{\rp}$ is $(\rp/3k)$-spam resistant, so that for all possible choices of $\trust_\graph{G} \subseteq \nodes_\graph{G}$ and  all possible spam graphs $\graph{H}$, $\mtermialgot{k}{\rp}(\graph{H},\trust_\graph{G})$ does not award a disproportionate amount of rank to the spammer. 

In this setting, T-PPR$_\rp$ is $\rp$-spam resistant but not $(2\rp)$-spam resistant, so we see that $\mtermialgot{k}{\rp}$ resists spam almost as well as $\mbox{T-PPR}_\rp$. Moreover, Theorem~\ref{thm:minppr-dist} implies that there are many possible choices of $\trust_\graph{G}$ such that $\mtermialgot{k}{\rp}(\graph{G},\trust_\graph{G})$ has $1+o(1)$ distortion, so that $\mtermialgot{k}{\rp}$ is accurate on the pre-spam graph. Thus $\mtermialgot{k}{\rp}$ performs far better than $\mbox{T-PPR}_\rp$, which can have distortion $\Omega(n)$ for all choices of $\trust_\graph{G}$, or even worst-case $\mbox{UPR}_\rp$, which can have distortion $\Omega(n^\dist)$.
Finally, \tminppr is a \pr by Lemma~\ref{lem:min-is-PR}. We therefore believe that \tminppr is a promising new ranking algorithm that warrants significant further study.

%% file: closure.tex
\section{\pr Closure}
\label{sec:closure}

In this section, we first introduce notation and make some preliminary
observations about \pr.  We then show that \pr is closed under normalized component-wise min.
Finally, we show that \pr is not closed under all functions, and in particular that it is not closed under median.

\subsection{\pr Preliminaries}


Recall that $(\graph{G},\R,\rp)$ is the random walk associated with PageRank on $\graph{G}$ with reset vector $\R$ and reset probability $\rp$. The transition probability matrix, $\mtrx{A}$, of this walk is
\begin{align}
\mtrx{A} = (1-\rp)\mtrx{M}+\rp \mtrx{R} , \label{eq:A}
\end{align}
where $\mtrx{M}$ and $\mtrx{R}$ denote $n \times n$ matrices as follows:
\begin{align*}
\forall u,v\in \nodes:\mtrx{M}[u,v]  &= \left\{ \begin{array}{ll}
 1/|\OUT(u)| & \textrm{if $(u,v)\in \edges$,}\\
 0 & \textrm{otherwise,}
  \end{array} \right.\\
\mtrx{R}[u,v] &= \R[v].
\end{align*}
For instance, if $\R$ is the uniform reset vector $\uniform$, then $\mtrx{R}[u,v] = 1/n$ for all $u,v \in \nodes$. 

Brin and Page noted that $\station{\graph{G},\uniform,\rp}$ is total, that is, it is defined and unique for any graph $\graph{G}$ and any $\rp\in(0,1)$.  We observe the following more general folklore lemma for \prs.  
\begin{lemma}
\label{thm:uniqStatDist}
$\station{\graph{G},\R,\rp}$ is defined and unique for any graph $\graph{G}=(\nodes,\edges)$,
reset vector $\R$, and reset probability $\rp\in (0,1)$.
\end{lemma}
\begin{proof}
Let $\nodes_{\R}$ be the support
of $\R$.
Notice that these nodes belong to a single strongly connected component in
the walk $(\graph{G},\R,\rp)$ 
consisting of the nodes reachable from $\nodes_{\R}$.  These
nodes form a unique essential communicating class\footnote{States $i$ 
and $j$ of a Markov chain belong to the same communicating class if 
there is a positive probability of moving to state $j$ from state $i$,
and a positive probability of moving to state $i$ from state $j$.}
in the Markov chain
of the random walk on $\mtrx{A}$.  By Proposition 1.26
in~\cite{levin2009markov}, such a Markov chain has a unique stationary
distribution.
\end{proof}
A similar claim was proved in~\cite{andersen2006local}, but for undirected graphs.  For directed graphs, as in our case, $\station{\graph{G},\R,\rp}$ has weight 0 on all nodes not reachable from a node in $\nodes_{\R}$. We will also need one more ancillary lemma.

\begin{lemma}\label{lem:stat-dist}
    Let $\graph{G} = (\nodes,\edges)$ be an arbitrary graph, let $0 < \rp < 1$, and let $\R$ be a reset vector. Let $(Y_t)_{t \ge 0}$ be the uniform random walk on $\graph{G}$ with random initial state drawn from $\R$. Then for all $v \in \nodes$, we have
    \[
        \station{\graph{G},\R,\rp}[v] = \rp\sum_{i=0}^\infty(1-\rp)^i\prob(Y_i = v).
    \]
\end{lemma}

\begin{proof}
    Follows from Equation 5 in~\cite{DBLP:conf/www/JehW03} and linearity of expectation. 
\end{proof}

For all $\rp \in (0,1)$, we denote  the set of all possible \prs for $\graph{G}$ with reset probability
$\rp$ by $\set{P}_\rp(\graph{G}) = \{\station{\graph{G},\R,\rp} \colon \R\mbox{ is a ranking vector}\}$. (Recall from Section~\ref{sec:def-spam-resistance} that a ranking vector on $\graph{G}=(\nodes,\edges)$ is any function $\vect{x}\colon \nodes \to [0,1]$ with $\sum_{v \in \nodes}\vect{x}[v] = 1$.)
We denote the set of all possible \prs for $\graph{G}$ with any reset probability by
$\set{P}(\graph{G}) = \{\station{\graph{G},\R,\rp} \colon \R\mbox{ is a ranking vector, }\rp \in (0,1)\}$.

We now set out a necessary and sufficient condition for a ranking vector to be a PageRank with a given reset probability. For all graphs $\graph{G}=(\nodes,\edges)$, all ranking vectors $\p$ on $\graph{G}$, and all $\rp \in (0,1)$, define
\[
    \stationinv{\graph{G},\vect{p},\rp}[v] := \frac{\vect{p}[v]}{\rp} - \frac{1-\rp}{\rp}\sum_{w \in \IN(v)} \frac{\vect{p}[w]}{\dout(w)} \mbox{ for all }v \in \nodes.
\]

\begin{lemma}\label{lem:PR-rp-cond}
    Let $\graph{G} = (\nodes,\edges)$ be a graph, let $\rp \in (0,1)$, and let $\vect{p}$ be a ranking vector on $\graph{G}$. If $\vect{p} = \station{\graph{G},\R,\rp}$ for some $\R$, then $\R = \stationinv{\graph{G},\vect{p},\rp}$. Moreover, $\vect{p} \in \set{P}_\rp(\graph{G})$ if and only if $\stationinv{\graph{G},\vect{p},\rp} \ge \mathbf{0}$.
\end{lemma}

\begin{proof}
    First suppose that $\vect{p} = \station{\graph{G},\R,\rp}$ for some $\R$. Let $\mtrx{M}$ be the transition matrix of a uniform random walk on $\graph{G}$ whose initial state is given by $\R$, and let $\mtrx{R}$ be the $|\nodes| \times |\nodes|$ matrix whose rows are given by $\R$. By definition, $\vect{p}$ is the unique (row) vector satisfying $\vect{p} = \vect{p}(\rp \mtrx{R}  + (1-\rp)\mtrx{M})$. Equivalently, $\vect{p}$ is the unique vector such that for all $v \in \nodes$,
    \begin{equation}\label{eqn:PR-rp-cond}
        \vect{p}[v] = \rp \sum_{w \in \nodes} \vect{p}[w]\R[v] + (1-\rp)\sum_{w \in \nodes}\vect{p}[w]\mtrx{M}[w,v] = \rp \R[v] + (1-\rp)\sum_{w \in \IN(v)}\frac{\vect{p}[w]}{\dout(w)}.
    \end{equation}
    Rearranging, we obtain $\R[v] = \stationinv{\graph{G},\vect{p},\rp}[v]$, and so $\R = \stationinv{\graph{G},\vect{p},\rp}$ as required. This also implies that if $\vect{p} \in \set{P}_\rp$, then $\stationinv{\graph{G},\vect{p},\rp} \ge \mathbf{0}$.
    
    Suppose now that $\stationinv{\graph{G},\vect{p},\rp} \ge \mathbf{0}$. We have
    \[
        ||\stationinv{\graph{G},\vect{p},\rp}|| = \frac{1}{\rp} - \frac{1-\rp}{\rp}\sum_{v \in \nodes}\sum_{w \in \IN(v)}\frac{\vect{p}[w]}{\dout(w)} = \frac{1}{\rp} - \frac{1-\rp}{\rp}\sum_{w \in \nodes}\sum_{v \in \OUT(w)}\frac{\vect{p}[w]}{\dout(w)} = \frac{1}{\rp} - \frac{1-\rp}{\rp} = 1,
    \]
    so $\stationinv{\graph{G},\vect{p},\rp}$ is a ranking vector on $\graph{G}$. Moreover, taking $\R = \stationinv{\graph{G},\vect{p},\rp}$, for all $v \in \nodes$ we have
    \[
        \rp \R[v] + (1-\rp)\sum_{w \in \IN(v)}\frac{\vect{p}[w]}{\dout(w)} = \vect{p}[v] - (1-\rp)\sum_{w \in \IN(v)}\frac{\vect{p}[w]}{\dout(w)} + (1-\rp)\sum_{w \in \IN(v)}\frac{\vect{p}[w]}{\dout(w)} = \vect{p}[v].
    \]
    Hence by~\eqref{eqn:PR-rp-cond}, we have $\vect{p} = \station{\graph{G},\R,\rp}$, and in particular $\vect{p} \in \set{P}_\rp(\graph{G})$.
\end{proof}

Next, using Lemma~\ref{lem:PR-rp-cond}, we set out a simple necessary and sufficient condition for a ranking vector to be a PageRank at all.

\begin{lemma}\label{lem:PR-cond}
    Let $\graph{G} = (\nodes,\edges)$ be a graph, and let $\vect{p}$ be a ranking vector on $\graph{G}$. Then $\vect{p} \in \set{P}(\graph{G})$ if and only if for all $(v,w) \in \edges$, if $\vect{p}[v] > 0$, then $\vect{p}[w] > 0$.
\end{lemma}
\begin{proof}
    Suppose $\vect{p} \in \set{P}(\graph{G})$ with $\vect{p} = \station{\graph{G},\R,\rp}$, let $v \in \nodes$, and suppose $\vect{p}[v] > 0$. Then for all $w \in \OUT(v)$, the PageRank random walk associated with $\vect{p}$ transitions from $v$ to $w$ with probability at least $(1-\rp)/|\OUT(v)| > 0$, so we must have $\vect{p}[w] > 0$.
    
    Conversely, let $\vect{p}$ be a ranking vector on $\graph{G}$, and suppose that $\vect{p}$ satisfies the condition that 
    for all $(v,w) \in \edges$, if $\vect{p}[v] > 0$, then $\vect{p}[w] > 0$. For all $v \in \nodes$, let
    \[
        \Sigma_v := \sum_{w \in \IN(v)}\frac{\vect{p}[w]}{|\OUT(w)|},\qquad\qquad\qquad
        x_v := \begin{cases}
            1 - \vect{p}[v]/\Sigma_v & \mbox{ if }\Sigma_v \ne 0,\\
            0 & \mbox{ otherwise,}
        \end{cases}
    \]
    and let $\rp = \max(\{1/2\} \cup \{x_v \colon v \in \nodes\})$. We now show that $\vect{p} \in \set{P}_\rp(\graph{G})$ by showing that $\rp \in (0,1)$, that $\stationinv{\graph{G},\vect{p},\rp} \ge \mathbf{0}$, and applying Lemma~\ref{lem:PR-rp-cond}.
    
    By definition, $\rp \ge 1/2 > 0$. For all $v \in \nodes$ with $\Sigma_v \ne 0$, there must exist $w \in \IN(v)$ with $\vect{p}[w] > 0$, so by hypothesis we have $\vect{p}[v] > 0$; hence $x_v < 1$. When $\Sigma_v = 0$ we have $x_v  = 0 < 1$ by definition, so it follows that $\rp < 1$; hence $\rp \in (0,1)$. 
    
    Now let $v \in \nodes$. If $\Sigma_v = 0$, then $\stationinv{\graph{G},\vect{p},\rp}[v] = \vect{p}[v]/\rp \ge 0$.
    If instead $\Sigma_v > 0$, then we have
    \[
        \stationinv{\graph{G},\vect{p},\rp}[v] \ge \frac{1}{\rp}\vect{p}[v] - \frac{1-x_v}{\rp}\Sigma_v  = 0.
    \]
    Thus $\stationinv{\graph{G},\vect{p},\rp} \ge \mathbf{0}$, so it follows by Lemma~\ref{lem:PR-rp-cond} that $\vect{p} \in \set{P}_\rp(\graph{G})$. In particular, $\vect{p} \in \set{P}(\graph{G})$ as required.
\end{proof}

In Lemma~\ref{lem:pr-enum}, we enumerate the possible realizations of a given ranking vector as a PageRank; before proving this, we introduce an ancillary lemma.

\begin{lemma}\label{lemma:epsuniq}
    Let $\graph{G} = (\nodes,\edges)$ be a graph, let $\p\in\set{P}(\graph{G})$, let $\R$ be a ranking vector on $\graph{G}$, and suppose there exists $\rp \in (0,1)$ such that $\p = \station{\graph{G},\R,\rp}$. Then exactly one of the following holds:
    \begin{enumerate}
        \item there exists $v \in \nodes$ such that $\R[v] \ne (\p\mtrx{M})[v]$, in which case 
        \[
            \rp = \frac{\p[v] - (\p\mtrx{M})[v]}{\R[v]-(\p\mtrx{M})[v]};\mbox{ or}
        \]
        \item $\R = \p = \p\mtrx{M}$, in which case $\p = \station{\graph{G},\R,\eta}$ for all $\eta\in(0,1)$.
    \end{enumerate}
\end{lemma}
\begin{proof}
    First suppose there exists $v \in \nodes$ such that $\R[v] \ne (\p\mtrx{M})[v]$. By Lemma~\ref{lem:PR-rp-cond}, we have
    \[
        \R[v] = \stationinv{\graph{G},\p,\rp}[v] = \frac{\p[v]}{\rp} - \frac{1-\rp}{\rp} (\p\mtrx{M})[v].
    \]
    Rearranging, we obtain
    \[
        \rp = \frac{\p[v] - (\p\mtrx{M})[v]}{\R[v]-(\p\mtrx{M})[v]},
    \]
    as required. 
    
    If no such $v \in \nodes$ exists, then we must have $\R = \p\mtrx{M}$. Since $\p$ is a PageRank, it follows that
    \[
        \p = \p\mtrx{A} = \p\big((1-\rp)\mtrx{M}+\rp\mtrx{R}\big) = (1-\rp)\R + \rp\R = \R,
    \]
    so $\R = \p$. Since $\R = \p\mtrx{M}$, it follows that $\R = \p = \p\mtrx{M}$ as required. Finally, for all $\eta \in (0,1)$, we have 
    \[
        \p\big((1-\eta)\mtrx{M}+\eta\mtrx{R}\big) = (1-\eta)\p\mtrx{M} + \eta\R = \p,
    \]
    so $\p = \station{\graph{G},\R,\eta}$ as required.
\end{proof}

\begin{lemma}\label{lem:pr-enum}
    Let $\graph{G} = (\nodes,\edges)$ be a graph, and let $\p \in \setP(\graph{G})$. Let $\mathbf{X}$ be the set of all pairs $(\R,\rp)$ such that $\p = \station{\graph{G},\R,\rp}$. If $\p\mtrx{M} = \p$, then we have $\mathbf{X} = \{(\p,\rp) \colon \rp \in (0,1)\}$; otherwise, we have
    \[
        \mathbf{X} = \{(\stationinv{\graph{G},\p,\rp},\rp) \colon \stationinv{\graph{G},\p,\rp} \ge \mathbf{0}\}.
    \]
\end{lemma}
\begin{proof}
    Immediate from Lemmas~\ref{lem:PR-rp-cond} and~\ref{lemma:epsuniq}.
\end{proof}

\subsection{\pr Closure Definitions}

Let $g$ be a function that takes a finite set of \prs on a graph $\graph{G}$ and returns a ranking vector for $\graph{G}$.  We say that $g$ is \emph{weakly closed on \prs} if, for any graph $\graph{G}$ and any finite set $\set{Q}\subset\set{P}(\graph{G})$, we have $g(\set{Q})\in\set{P}(\graph{G})$.

The reason we call this type of closure ``weak'' is that by Lemma~\ref{lem:PR-cond}, 
being a \pr\ imposes only a mild condition satisfied by e.g.\ any vector with no zero entries. The reason this mild condition suffices is that the reset probability, $\epsilon$, may be arbitrarily close to~$1$.
If \pr is to capture any graph structure, as opposed to simply approximating the reset vector, then $\rp$ needs to be well below, say, 1/2.  

Therefore, we define $g$ to be \emph{strongly closed on \prs} if, for all graphs $\graph{G}$, all $\rp \in (0,1)$, and all finite  sets $\set{Q}\subset\set{P}_\rp(\graph{G})$, we have $g(\set{Q})\in\set{P}_\rp(\graph{G})$.  That is, strong closure implies that the operator produces a \pr without changing $\rp$. 

In the following two sections, we demonstrate that min is strongly closed for \prs, whereas median is only weakly closed. 

\subsection{Strongly Closed Operators}

In this section, we show that \pr is strongly closed under the normalized component-wise $\min$ operator and prove Theorem~\ref{thm:closedmin}, which implies Theorem~\ref{lem:min-is-PR}.

\begin{theorem}\label{thm:closedmin}
    Let $\graph{G} = (\nodes,\edges)$ be a graph, let $\rp \in (0,1)$, and suppose $\vect{x}_1, \dots, \vect{x}_{k} \in \set{P}_\rp(\graph{G})$. For all $v \in \nodes$, let $\scaledvect{y}[v] = \min\{\vect{x}_1, \dots, \vect{x}_k\}[v]$, suppose $\scaledvect{y} \ne \mathbf{0}$, and let $\vect{y} = \scaledvect{y}/||\scaledvect{y}||$\,. Then $\vect{y} \in \setP_\rp(\graph{G})$.
\end{theorem}
\begin{proof}
    First note that $\vect{y}$ is a ranking vector for $\graph{G}$. Let $v \in \nodes$. Then for all $i \in [k]$, since $\vect{x}_i \in \setP_\rp(\graph{G})$, by Lemma~\ref{lem:PR-rp-cond} we have
    \begin{equation}\label{eqn:closedmin}
        \stationinv{\graph{G},\vect{x}_i,\rp}[v] = \frac{\vect{x}_i[v]}{\rp} - \frac{1-\rp}{\rp} \sum_{w \in \IN(v)}\frac{\vect{x}_i[w]}{\dout(w)} \ge \mathbf{0}.
    \end{equation}
    For all $v \in \nodes$, we have $\scaledvect{y}[v] = \vect{x}_j[v]$ for some $j\in [k]$, and $\scaledvect{y}[w] \le \vect{x}_j[w]$ for all $w \in \IN(v)$. It follows from~\eqref{eqn:closedmin} that
    \[
        \frac{\scaledvect{y}[v]}{\rp} - \frac{1-\rp}{\rp} \sum_{w \in \IN(v)}\frac{\scaledvect{y}[w]}{\dout(w)} \ge \mathbf{0},
    \]
    and hence $\stationinv{\graph{G},\vect{y},\rp}[v] \ge 0$. Thus $\stationinv{\graph{G},\vect{y},\rp} \ge \mathbf{0}$, so by Lemma~\ref{lem:PR-rp-cond} we have $\vect{y} \in \setP_\rp(\graph{G})$.
\end{proof}

Recall from Section~\ref{sec:results} that for any graph $\graph{G} = (\nodes,\edges)$, any $\rp \in (0,1)$, and any coherent set $\centers \subseteq \nodes$, we have $\min\{\station{\graph{G},x,\rp}\colon x \in \centers\} \ne \mathbf{0}$. Thus Theorem~\ref{thm:closedmin} implies that $\stationmin{\graph{G},\centers,\rp} \in \setP_\rp(\graph{G})$, so Lemma~\ref{lem:min-is-PR} follows.

\subsection{Weakly Closed Operators}

In this section we show that median is weakly closed, but not strongly closed, over \prs. 
We first define the median operator on vectors by
component-wise median, so that for vectors $\scaledvect{x_1}, \ldots, \scaledvect{x_k}$, 
\[
\textnormal{median}\{\scaledvect{x_1}, \ldots, \scaledvect{x_k}\}[i] = \textnormal{median}\{\scaledvect{x_1}[i], \ldots, \scaledvect{x_k}[i]\}.
\]
We now formally define the median operator on ranking vectors.
The normalized component-wise median of \prs may not be well-defined as the median might be identically zero. Therefore, we add a condition to the definition to avoid those cases.
\begin{definition}
	Let $\graph{G} = (\nodes, \edges)$ be a graph 
	and let $\prset = \{\x_1,\dots, \x_k\}$ such that $\x_i \in \set{P}(\graph{G})$ for every $i=1,\dots,k$ and 
	$||\textnormal{median}\{\scaledvect{\x_i} \colon \x_i\in \prset\}||>0$. Then, we define the Median operator as 
\[
	 \medtermi{\graph{G},\prset} =
	 \norm{\textnormal{median}\{\scaledvect{\x_i} \colon \x_i\in \prset\}}
\]
\end{definition}

The following theorem says that not only is the median operator not strongly closed, but that the property fails badly --- in general, we cannot express the normalized component-wise median of even low-reset-probability PPRs as a PageRank without using a reset probability greater than $1/2$.

\begin{lemma}\label{lem:median-bad}
    Let $\rp \in (0,1)$. Then there exist infinitely many graphs $\graph{G} = (\nodes,\edges)$ and sets of \PPRs $\prset = \{\x_1,\dots,\x_k\}$ such that $\x_1,\dots,\x_k \in \setP_\rp(\graph{G})$ and 
   $ ||\textnormal{median}\{\x_1,\dots,\x_k\}||>0$, but $\medtermi{\graph{G},\prset} \notin \setP_\rp(\graph{G})$. Moreover, $\medtermi{\graph{G},\prset} \notin \setP_\eta(\graph{G})$ for any $\eta \le 1/2$.
\end{lemma}
\begin{proof}
    Let $k$ be any odd integer satisfying $k \ge 3$ and $k > (1-\rp)/\rp$, and write $k=:2\ell+1$. Consider the graph in Figure~\ref{fig:medianpr}, where each node $u_i$ is connected to the $\ell+1$ nodes $v_i, v_{(i+1)\mymod k},\dots,v_{(i+\ell)\mymod k}$. For all $i \in [k]$, we define $\x_i := \station{\graph{G},u_i,\rp}$. Thus the reset vector $\R_{\x_i}$ of $\x_i$ satisfies $\R_{\x_i}[u_i] = 1$ and $\R_{\x_i}[v] = 0$ for all $v \ne u_i$. By Lemma~\ref{lem:stat-dist},
    we have
    \begin{align}\label{eqn:med-only-weak}
        \x_i[v_j] = \begin{cases}
            \rp(1-\rp)/(\ell+1) & \mbox{ if }u_i \in \IN(v_j),\\
            0 & \mbox{otherwise,}
        \end{cases}\qquad\qquad
        \x_i[y_1] = \rp(1-\rp)^2.
    \end{align}
    For brevity, write $a = ||\textnormal{median}\{\x_1,\dots,\x_k\}||$, and note that $a>0$. Since $\ell+1 > k/2$, it follows from~\eqref{eqn:med-only-weak} that
    \begin{align*}
        \medtermi{\graph{G},\prset}[v_j] &= \rp(1-\rp)/(a(\ell+1)) \mbox{ for all }j \in [k],\\
        \medtermi{\graph{G},\prset}[y_1] &= \rp(1-\rp)^2/a.
    \end{align*}
    
    Now let $\eta \in (0,1)$. 
    By Lemma~\ref{lem:PR-rp-cond}, 
    whenever $\mathcal{R}^{-1}(\graph{G},\medtermi{\graph{G},\prset},\eta)[y_1] < 0$,
    we have $\medtermi{\graph{G},\prset} \notin \setP_\eta(\graph{G})$. 
    By the definition of $\mathcal{R}^{-1}$, we have
    \begin{equation*}
        \mathcal{R}^{-1}\big(\graph{G},\medtermi{\graph{G},\prset},\eta\big)[y_1] = \frac{\rp(1-\rp)^2}{\eta a} - \frac{1-\eta}{\eta}k\frac{\rp(1-\rp)}{a(\ell+1)} = \frac{\rp(1-\rp)}{\eta a}\Big(1-\rp-\frac{(1-\eta)k}{\ell+1}\Big).
    \end{equation*}
    Thus $\medtermi{\graph{G},\prset} \notin \setP_\eta(\graph{G})$ whenever $1-\eta > (1-\rp)(\ell+1)/k$, which holds if and only if
    \begin{equation}\label{eqn:med-bad}
        \eta < \frac{k-\ell-1+\rp\ell+\rp}{k} =
        \frac{1+\rp}{2} - \frac{1-\rp}{2k}.
    \end{equation}
    Since $k \ge 3$,~\eqref{eqn:med-bad} holds when $\eta = \eps$, and since $k > (1-\rp)/\rp$,~\eqref{eqn:med-bad} holds for all $\eta \le 1/2$. The result therefore follows.
\end{proof}


\begin{figure}[htbp]
\centering
\includegraphics[width=0.8\textwidth]{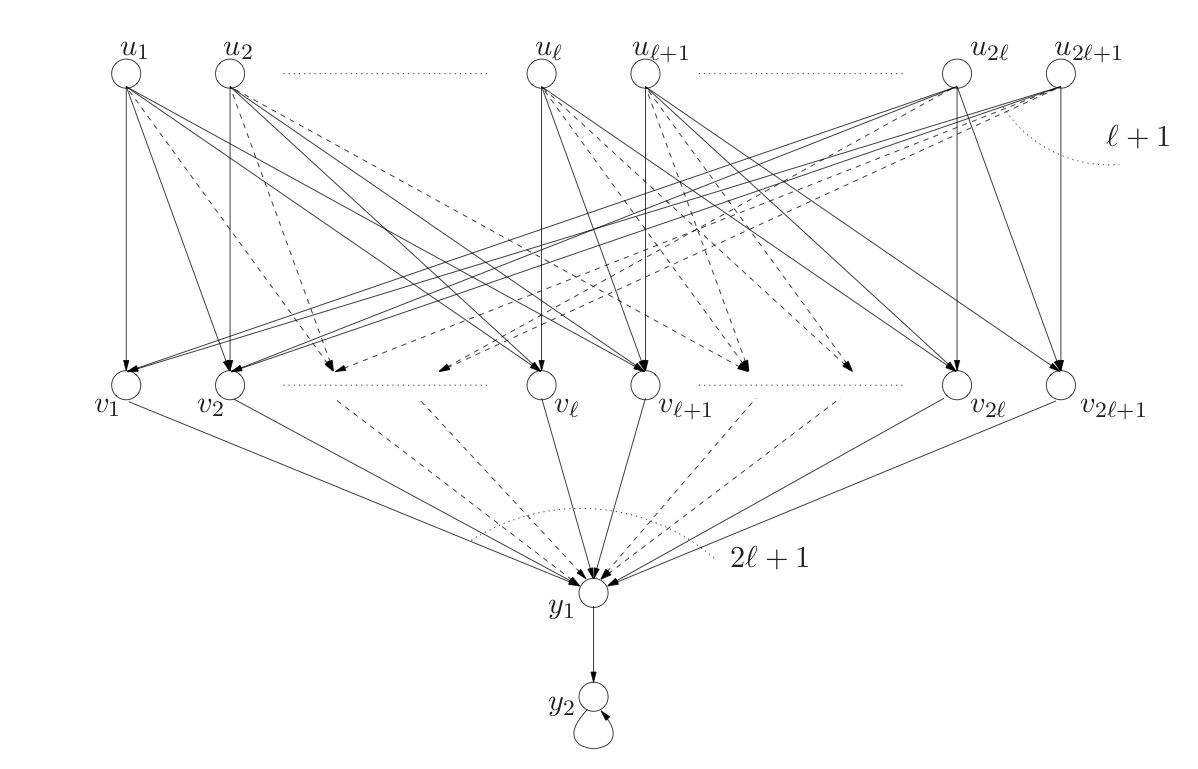}
\caption{Illustration of weak closure of median.}
\label{fig:medianpr}
\end{figure}

Finally we show that Median is weakly closed over \prs. Thus taking a median of \prs always yields a \pr, but perhaps one with a much higher reset probability.

\begin{lemma}\label{lem:median-weak-closure}[Median is weakly closed over \prs]
    Let $\graph{G} = (\nodes,\edges)$ be a graph and let $\prset = \{\x_1, \dots, \x_{2\ell+1}\}$, where $\x_i \in \set{P}(\graph{G})$ for every $i=1,\dots,2\ell+1$. Suppose
    $||\textnormal{median}\{\scaledvect{\x_i} \colon \x_i\in \prset\}||>0$.
    Then, $\medtermi{\graph{G},\prset} \in \setP(\graph{G})$.
\end{lemma}

\begin{proof}
    Let $\medtermi{\nodes,\prset}=\vect{y}$, and let $v \in \nodes$ with $\vect{y}[v] > 0$. Then $\scaledvect{y}[v] > 0$, so there exist $i_1, \dots, i_{\ell+1}$ with $\vect{x}_{i_1}[v],\dots,\vect{x}_{i_{\ell+1}}[v] > 0$. Since $\vect{x}_1,\dots,\vect{x}_{2\ell+1} \in \setP(\graph{G})$, by Lemma~\ref{lem:PR-cond} it follows that for all $w \in \OUT(v)$, we have $\vect{x}_{i_1}[w],\dots,\vect{x}_{i_{\ell+1}}[w] > 0$. Hence by construction we have $\scaledvect{y}[w] > 0$ and therefore $\vect{y}[w] > 0$ for all $w \in \OUT(v)$. It follows by Lemma~\ref{lem:PR-cond} that $\vect{y} \in \set{P}(\graph{G})$.
\end{proof}

%% file: related.tex

\section{Related Work}
\label{sec:related}

Originally, and most famously, \pr was used by Google as a ranking function for web pages~\cite{GooglePr31:online}, but since then, it has been used to analyze networks of neurons~\cite{fletcher2016structure}, Twitter recommendation systems
~\cite{DBLP:conf/www/GuptaGLSWZ13},
protein networks~\cite{ivan2010web}, etc. (See~\cite{gleich2015pagerank} for a survey of non-web uses). 

As noted above, \pr is susceptible to link spam.  Thus, other ranking functions have been proposed~\cite{kumar2006core, bhattacharjee2006incentive, gyongyi2004combating}.
TrustRank~\cite{gyongyi2004combating} for example is based on assigning higher reputation to a subset of pages curated by an expert, and the assumption that pages linked from these reputable pages are reputable as well. A similar method can be applied for low reputation pages, which is called Anti-Trust Rank~\cite{krishnan2006web}. In both, reliability lowers as distance from the reference pages increases.

Other work is geared towards modifications of the \pr mechanism. For instance, Global Hitting Time~\cite{hopcroft2008manipulation} was designed as a transformation of \pr to counter cross-reference link spam, where nodes link each other to increase their rank, but it still suffers if the number of spammers is large. Variants include Personalized Hitting Time~\cite{liu2016personalized}.

Despite the progress on other ranking mechanisms, \pr still stands as the most popular~\cite{searchEngineMarketShare} ranking function, and therefore the most attractive for link-spammers.
Google discouraged \pr manipulation through the buying of highly ranked links by social methods: they have announced that pages discovered to participate in such activity will be left out of the \pr calculation (hence, their rank lowered), they have encouraged the public to notify Google about such pages~\cite{googleSpammerNotifLink}.

%
%
%
%

Other research has focused on link-spam detection~\cite{gyongyi2006link} and quantifying the rank increase obtained by creating Sybil pages~\cite{cheng2006manipulability}. 
For instance, an algorithm to detect spam analyzing the supporting sets, i.e. the sets of nodes that contribute the most to the \pr of a given vertex, was presented in~\cite{DBLP:conf/airweb/AndersenBCHJMT08}. The performance evaluation is experimental. Detection methods for Sybil pages attacks have been surveyed in~\cite{Yu:2011:SDV:2034575.2034593,alvisi2014communities}. Some of those methods~\cite{DBLP:journals/ton/YuGKX10,yu2008sybilguard} are based on detecting abnormal random walk mixing times for what is expected in an ``honest'' network. 
Link-spam detection may be useful for excluding pages from the \pr calculation, but it is better to render an attack futile than to build a fortress.  That is, it is better to develop techniques that yield \pr spam resistant.  Towards that end, some work limits or assign reset probability selectively~\cite{fogaras2005towards,gyongyi2004combating}. These approaches are generalizations of Personalized \pr~\cite{DBLP:conf/www/JehW03}.

For graph-theoretic ranking functions, such as Hubs \& Authorities (HITS) and \pr, formalizations of how \emph{stable} they are in face of small perturbations exist~\cite{Ng:2001:SAL:383952.384003,DBLP:conf/ijcai/NgZJ01}.
Stability refers to how sensitive eigenvector methods such as HITS and \pr are to small changes in the link structure, and the cost of introducing such perturbations is not considered.
Spammability, on the other hand, is a different metric because it relates the cost of perturbations to the gains obtained by those introducing them.

Specifically for HITS on a graph with adjacency matrix $A$, the authors in~\cite{Ng:2001:SAL:383952.384003,DBLP:conf/ijcai/NgZJ01} relate an upper bound on the number of links that may be added (or deleted), given as a function of the maximum out-degree and the eigengap of $A^TA$, to an upper bound on the change of the principal eigenvector of $A^TA$ that those link changes produce. The result characterizes stability because HITS uses the principal eigenvector of $A^TA$ to determine authorities. The authors also show the existence of graphs where a small perturbation (e.g. adding a single link) has a large effect.

For uniform \pr, in~\cite{Ng:2001:SAL:383952.384003,DBLP:conf/ijcai/NgZJ01} they upper bound the aggregated change in rank over all pages ($\ell_1$-norm) 
by a $2/\rp$ factor of the sum of the original rank of the pages whose out-links were changed, where $\rp$ is the reset probability.
Considering rank as a measure of cost, this result can be seen as relating the overall impact on the system to the cost of introducing changes, but it does not relate to the increase in rank for those nodes. That is, it characterizes stability but not spammability. Moreover, uniform \pr under this cost measure can be spammed for free by simply creating new nodes, and non-uniform reset vectors are not considered.
%

As expected, personalized \pr is biased towards the vicinity of the trusted node. This undesired effect can be compensated for to some extent by concentrating reset probability on a subset of nodes rather than one (as in~\cite{fogaras2005towards,gyongyi2004combating}). Indeed, the approach has been successful for particular areas where the search space is relatively small (e.g. in Linguistic Knowledge Builder graph~\cite{agirre2009personalizing}, Social Networks~\cite{bahmani2010fast,jin2012lbsnrank}, and Protein Interaction Networks~\cite{ivan2010web}). But the scale of the web graph may require a large set of trusted pages for a general purpose \pr.  


%% file: distortion.tex
\section{Distortion of \minppr}
\label{sec:distortion}


We first set out notation for mixing times from specific initial states. Let $\graph{G} \in \ergodicG$, let $\R$ be a probability distribution on $\nodes$, let $X \sim \R$, and let $\vect{p}_{i,\R}$ be the distribution of the uniform random walk on $\graph{G}$ at time $i \ge 0$ from initial state $X$. Then for all $\rho>0$, we define
\begin{align*}
    \mix_\graph{G}(\rho,\R) &:= \min\Big\{i \ge 0 \colon \dtv\big(\vect{p}_{i,\R},\, \station{\graph{G}}\big) \le \rho\Big\}.
\end{align*}
In the special case where $\R$ is deterministic, i.e.\ there exists $x \in \nodes$ such that $\R[x] = 1$, we write $\mix_\graph{G}(\rho,x) := \mix_\graph{G}(\rho,\R)$. We take the default value of $\rho$ to be $1/4$, so that $\mix_\graph{G} := \mix_\graph{G}(1/4)$,
$\mix_\graph{G}(x) := \mix_\graph{G}(1/4,x)$ and $\mix_\graph{G}(\R) := \mix_\graph{G}(1/4,\R)$. 
We also define $\emix_\graph{G}(\R) := \mean_{X \sim \R}(\tau_\graph{G}(X))$, where ``emt'' stands for ``expected mixing time''.  For all positive integers $k$ we will write $[k] = \{1, \dots, k\}$. 

We now state some well-known preliminary lemmas. 

\begin{lemma}[{\cite[Eq 10]{haveliwala2003topic}}]\label{lem:lin-comb-PR}
	Let $\graph{G}$ be an arbitrary graph, and let $\R$ be a reset vector on $\graph{G}$. Then for all $y \in \nodes_\graph{G}$, $\mPR{\R}[y] = \sum_{x \in \nodes_\graph{G}} \R[x]\mPR{x}[y]$.
\end{lemma}

\begin{lemma}\label{lem:mixing-time}
	Let $\graph{G} \in \ergodicG$, let $\rho \in (0,1)$, and let $\R$ be a reset vector on $\graph{G}$. Then $\tau_\graph{G}(\rho,\R) \le \ceil{\log_2(1/\rho)}\tau_\graph{G}(\R)$.
\end{lemma}
\begin{proof}
	This is immediate from~\cite[Theorem 11.6]{mitzenmacher2005book}, taking $P$ to be the transition matrix of the random walk associated with $\station{\graph{G},\R,\rp}$ and $c$ to be $1/4$.
\end{proof}

\begin{lemma}[{\cite[Theorem~4.4]{mitzenmacher2005book}}]\label{lem:new-chernoff}
	Let $X$ be a binomial random variable with mean $\mu$, and let $0 < \eta \le 1$. Then
	\[
		\prob(X \ge (1+\eta)\mu) \le e^{-\eta^2\mu/3}.\pushQED{\qed}\qedhere\popQED
	\]
\end{lemma}

\subsection{All \prs are close in total variation distance}

In this section, we prove Lemma~\ref{lem:contr-low} (which we will use later in the proof of Theorem~\ref{thm:minppr-dist}) and bound the total variation distance between any PageRank and the corresponding reference rank.

\begin{lemma}\label{lem:new-PR-high}
    Let $\graph{G} = (\nodes, \edges)$ be a graph in $\ergodicG$, let $\R$ be a reset vector on $\graph{G}$, and let $0 < \rp < 1$. Then for all $y \in \nodes$ with $\station{\graph{G}}[y] \ne 0$,
    \[
        \frac{\station{\graph{G}}[y] - \station{\graph{G},\R,\rp}[y]}{\station{\graph{G}}[y]} \le \rp\mix_\graph{G}(\R)\big(3 - \log_2\station{\graph{G}}[y]\big).
    \]
\end{lemma}
\begin{proof}
    Fix $y \in \nodes$, and for brevity define $\tau := \mix_\graph{G}(\R)$ and $I := \floor{{-}\log_2 \station{\graph{G}}[y]}+ 2$. Let $(M(t))_{t \ge 0}$ be the uniform random walk on $\graph{G}$ with initial state drawn from $\R$. Then by Lemma~\ref{lem:stat-dist}, we have
    \[
        \station{\graph{G},\R,\rp}[y] = \rp\sum_{t=0}^\infty (1-\rp)^t \prob\big(M(t) = y\big) \ge \rp\sum_{i=I}^\infty \sum_{t=0}^{\tau-1} (1-\rp)^{i\tau+t}\prob\big(M(i\tau + t) = y\big).
    \]
    By Lemma~\ref{lem:mixing-time}, for all $i \ge 1$ we have $\mix_\graph{G}(2^{-i}, \R) \le i\tau$. Thus by the definition of a mixing time, it follows that
    \begin{equation}\label{eqn:low-cont-1}
        \station{\graph{G},\R,\rp}[y] \ge \rp\sum_{i=I}^\infty \sum_{t=0}^{\tau-1} (1-\rp)^{i\tau+t}\big(\station{\graph{G}}[y] - 2^{-i}\big).
    \end{equation}
    
    We now split~\eqref{eqn:low-cont-1} into two terms and bound each separately. We have
    \begin{equation}\label{eqn:low-cont-2}
        \rp\sum_{i=I}^\infty\sum_{t=0}^{\tau-1}(1-\rp)^{i\tau+t}\station{\graph{G}}[y] = \rp(1-\rp)^{I\tau}\station{\graph{G}}[y]\sum_{t=0}^\infty(1-\rp)^t = (1-\rp)^{I\tau}\station{\graph{G}}[y] \ge (1-\rp I\tau)\station{\graph{G}}[y].
    \end{equation}
    Moreover, 
    \begin{align*}
        \rp\sum_{i=I}^\infty\sum_{t=0}^{\tau-1}(1-\rp)^{i\tau+t} 2^{-i} &= \rp\cdot \sum_{i=I}^\infty \big((1-\rp)^\tau/2\big)^i\cdot \sum_{t=0}^{\tau-1}(1-\rp)^t \le \rp\cdot \sum_{i=I}^\infty 2^{-i} \cdot \frac{1-(1-\rp)^\tau}{\rp}\\
        &\le \rp\cdot 2^{-I+1} \cdot \frac{1-(1-\rp \tau)}{\rp} \le \rp \cdot \station{\graph{G}}[y] \cdot \tau.
    \end{align*}
    (Here the last inequality relies on the definition of $I$.) 
    It follows from~\eqref{eqn:low-cont-1} and~\eqref{eqn:low-cont-2} that 
    \[
        \station{\graph{G},\R,\rp}[y] \ge \big(1-\rp (I+1) \tau\big)\station{\graph{G}}[y] \ge \big(1 - \rp\tau(3 - \log_2\station{\graph{G}}[y])\big)\station{\graph{G}}[y],
    \]
    so the result follows.
\end{proof}

{\def\thetheorem{\ref{lem:contr-low}}
\begin{lemma}[restated]
    \statecontrlow
\end{lemma}
\addtocounter{theorem}{-1}
}
\begin{proof}
    Let $y \in \nodes_\graph{G}$. If $\station{\graph{G}}[y] \le 1/n^\dist$, then $\cont_\dist(\station{\graph{G},\R,\rp},y) \le 1$, so suppose $\station{\graph{G}}[y] > 1/n^\dist$; thus we have $\cont_\dist(\station{\graph{G},\R,\rp},y) \le \station{\graph{G}}[y]/\station{\graph{G},\R,\rp}[y]$. By Lemma~\ref{lem:new-PR-high},
  using the fact that $\graph{G} \in \mixclass{T}$ (as defined in Section~\ref{sec:def-distortion}),
    it follows that
    \begin{align*}
        \cont_\dist(\station{\graph{G},\R,\rp},y) &\le \frac{1}{1 - \rp\mix_\graph{G}(\R)\big(3 -\log_2\station{\graph{G}}[y]\big)} \le \frac{1}{1 - \rp T(n)(3 + \dist\log_2 n)}.
    \end{align*}
    Since $T(n) \le 1/(2\rp(3+\dist\log_2 n))$, it follows that $\cont_\dist(\station{\graph{G},\R,\rp},y) \le 1 + 2\rp T(n)(3+\dist\log_2 n)$ as required.
\end{proof}

For any vector $\hat{\bf x}\colon V \to [0,1]$ let $H(\hat{\bf x})$ be its Shannon entropy,
namely $H(\hat{\bf x})= -\sum_{v \in V} \hat{\bf x}[v] \log_2 \hat{\bf x}[v]$.
The following theorem bounds the total variation distance between 
$\station{\graph{G},\R,\rp}$ and $\station{\graph{G}}$ in terms of the Shannon entropy of $\station{\graph{G}}$.

\begin{theorem}\label{thm:dtv-low}
    Let   $\rp \in (0,1)$.  Then for all $n$-vertex graphs $\graph{G} \in \mixclass{}$ and all reset vectors $\R$ on $\graph{G}$, the total variation distance between $\station{\graph{G},\R,\rp}$ and $\station{\graph{G}}$ is at most $\rp\mix_\graph{G}(\R)(3 + H(\station{\graph{G}}))$.
\end{theorem}

For any 
$\graph{G} \in \mixclass{}$, $\tau_{\graph{G}}$ is the mixing time of $\graph{G}$ from an arbitrary vertex 
and
$H(\station{\graph{G}}) \leq \log_2 n$, so
Theorem~\ref{thm:dtv-low} implies that on an $n$-vertex graph $\graph{G} \in \ergodicG$, any PageRank with reset probability $\rp$ has total variation distance at most $\mix_\graph{G}\rp(3+\log_2 n)$ to the reference rank of $\graph{G}$.

\begin{proof}
    We have
    \begin{equation}\label{eqn:dtv-low}
        \dtv\big(\station{\graph{G},\R,\rp},\,\station{\graph{G}}\big) = \frac{1}{2}\sum_{y \in \nodes_\graph{G}} \big|\station{\graph{G},\R,\rp}[y] - \station{\graph{G}}[y] \big|.
    \end{equation}
    For all $y \in \nodes_\graph{G}$, let
    \[
        \delta_y := \begin{cases}
            \big(3 - \log_2\station{\graph{G}}[y]\big)\station{\graph{G}}[y] & \mbox{ if }\station{\graph{G}}[y] \ne 0,\\
            0 & \mbox{ otherwise}.
        \end{cases}
    \] 
    By Lemma~\ref{lem:new-PR-high}, each vertex $y \in \nodes_\graph{G}$ with $\station{\graph{G},\R,\rp}[y] < \station{\graph{G}}[y]$ contributes at most $\rp\tau_{\graph{G}}(\R)\delta_y/2$ to the sum in~\eqref{eqn:dtv-low}; thus in total such vertices contribute at most $\rp\tau_{\graph{G}}(\R)\sum_y \delta_y/2$. Moreover, the total contribution of all vertices $y \in \nodes_\graph{G}$ with $\station{\graph{G},\R,\rp}[y] > \station{\graph{G}}[y]$ is exactly the same. Thus in total,
    \begin{align*}
        \dtv\big(\station{\graph{G},\R,\rp},\,\station{\graph{G}}\big) &\le \rp\tau_{\graph{G}}(\R)\sum_{y \in \nodes_\graph{G}} \delta_y = \rp \tau_\graph{G}(\R) \big(3 + H\big(\station{\graph{G}}\big)\big),
    \end{align*}
    as required.
\end{proof}

\subsection{Min-PPR  can approximate $\station{\graph{G}}$ well everywhere} 
	
We first prove probabilistic bounds on the relative error of PPR when its center vertex is chosen randomly according to the reference rank of the graph; these are Lemmas~\ref{lem:new-minppr-upper} and~\ref{lem:new-minppr-lower}.

\begin{lemma}\label{lem:new-minppr-upper}
    Let $\rp \in (0,1)$. Let $\graph{G}$ be a graph in $\ergodicG$, and let $X$ be a vertex chosen randomly from $\nodes_\graph{G}$ according to $\station{\graph{G}}$. Then for all $y \in \nodes_\graph{G}$ with $\station{\graph{G}}[y] \ne 0$, with probability at least $7/8$, 
    \[
        \station{\graph{G},X,\rp}[y] \le \Big(1 + 8\rp\,\emix_\graph{G}(\station{\graph{G}})\big(3 - \log_2\station{\graph{G}}[y]\big)\Big)\station{\graph{G}}[y].
    \]
\end{lemma}
\begin{proof}
    Let $\delta := 8\rp\,\emix_\graph{G}(\station{\graph{G}})(3 - \log_2\station{\graph{G}}[y])$ for brevity, 
    and let $\mathbf{S}$ be the set of all vertices $X$ satisfying $\station{\graph{G},X,\rp}[y] \ge (1+\delta)\station{\graph{G}}[y]$. We will prove the lemma by showing that $\station{\graph{G}}[\mathbf{S}] \le 1/8$.
    
    Observe that $\station{\graph{G},\station{\graph{G}},\rp} = \station{\graph{G}}$, since $\station{\graph{G}}$ remains invariant under both resetting to $\station{\graph{G}}$ and uniformly random steps on $\graph{G}$. By Lemma~\ref{lem:lin-comb-PR}, applied with $\R=\station{\graph{G}}$, it follows~that
	\begin{equation}\label{eqn:new-prob-cover-1}
		\station{\graph{G}}[y] = \sum_{X \in \nodes_\graph{G}}\station{\graph{G}}[X]\cdot \station{\graph{G},X,\rp}[y].
	\end{equation}
	We split the sum in~\eqref{eqn:new-prob-cover-1} into two parts. By the definition of $\mathbf{S}$, we have
	\begin{equation}\label{eqn:new-prob-cover-2}
	    \sum_{X \in \mathbf{S}}\station{\graph{G}}[X]\cdot \station{\graph{G},X,\rp}[y] \ge \station{\graph{G}}[\mathbf{S}]\cdot(1+\delta)\station{\graph{G}}[y].
	\end{equation}
	Moreover, by Lemma~\ref{lem:new-PR-high} (applied with $\R = X$) we have
	\begin{align*}
	    \sum_{X \in \nodes_\graph{G} \setminus \mathbf{S}} \station{\graph{G}}[X]\cdot\station{\graph{G},X,\rp}[y] 
	    &\ge \sum_{X \in \nodes_\graph{G} \setminus \mathbf{S}} \station{\graph{G}}[X]\, \Big(1 - \rp \mix_\graph{G}(X)\big(3 - \log_2  \station{\graph{G}}[y]\big)\Big)\,\station{\graph{G}}[y]\\
	    &\geq \Big(\station{\graph{G}}[\nodes_\graph{G} \setminus \mathbf{S}] - \rp\,\emix_\graph{G}(\station{\graph{G}})\big(3 -\log_2\station{\graph{G}}[y]\big) \Big)\,\station{\graph{G}}[y].
	\end{align*} 
	It follows by~\eqref{eqn:new-prob-cover-1} and~\eqref{eqn:new-prob-cover-2} that
	\[
	    \station{\graph{G}}[y] \ge \Big(1 + \delta\station{\graph{G}}[\mathbf{S}] - \rp\,\emix_\graph{G}(\station{\graph{G}})\big(3 -\log_2\station{\graph{G}}[y]\big) \Big)\station{\graph{G}}[y],
	\]
	so by rearranging we obtain
	\[
	    \station{\graph{G}}[\mathbf{S}] \le \frac{1}{\delta}\cdot \rp\,\emix_\graph{G}(\station{\graph{G}})\big(3-\log_2\station{\graph{G}}[y]\big) = \frac{1}{8}
	\]
	as required.
\end{proof}

\begin{lemma}\label{lem:new-minppr-lower}
    Let $\rp \in (0,1)$. Let $\graph{G}$ be a graph in $\ergodicG$, and let $X$ be a vertex chosen randomly from $\nodes_\graph{G}$ according to $\station{\graph{G}}$. Then with probability at least $7/8$, for all $y \in \nodes_\graph{G}$ with $\station{\graph{G}}[y] \ne 0$,
    \[
       \station{\graph{G},X,\rp}[y] \ge \Big(1 - 8\rp\,\emix_\graph{G}(\station{\graph{G}})\big(3 - \log_2\station{\graph{G}}[y]\big)\Big)\station{\graph{G}}[y].
    \]
\end{lemma}

Note that while Lemma~\ref{lem:new-minppr-upper} applies to only a single vertex $y \in \nodes_\graph{G}$, Lemma~\ref{lem:new-minppr-lower} applies collectively to all such vertices.

\begin{proof}
    By Markov's inequality, we have
    \begin{equation}\label{eqn:new-rand-center-lower}
        \prob\big(\mix_\graph{G}(X) \ge 8\emix_\graph{G}(\station{\graph{G}})\big) \le 1/8.
    \end{equation}
    Suppose 
   $ \mix_\graph{G}(X) < 8\emix_\graph{G}(\station{\graph{G}})$.
   Now consider a vertex $y \in \nodes_\graph{G}$
   with
   $\station{\graph{G}}[y] \ne 0$. By Lemma~\ref{lem:new-PR-high}, applied with $\R = X$, we have
    \[
        \frac{\station{\graph{G}}[y] - \station{\graph{G},X,\rp}[y]}{\station{\graph{G}}[y]} \le \rp\mix_\graph{G}(X)\big(3 - \log_2\station{\graph{G}}[y]\big) < 8\rp\,\emix_\graph{G}(\station{\graph{G}})\big(3 - \log_2\station{\graph{G}}[y]\big).
    \]
    The result therefore follows from~\eqref{eqn:new-rand-center-lower}.
\end{proof}

We are now in a position to prove a general error bound for Min-PPR with a randomly-chosen set, from which Theorem~\ref{thm:minppr-dist} will follow easily.

\begin{lemma}\label{lem:minppr-easy-accuracy}
    Let $\rp \in (0,1)$, and let $\graph{G}$ be an $n$-vertex graph in $\ergodicG$. Suppose $\tau_\graph{G} \le 1/(2\rp(3+H(\station{\graph{G}})))$. Let $\vect{p}$ be a probability distribution on $\nodes_\graph{G}$ with $\dtv(\vect{p},\station{\graph{G}}) \le 1/8$. Let $X_1, \dots, X_k \sim \vect{p}$ be independent and identically distributed, and let $\centers = \{X_1, \dots, X_k\}$. Then with probability at least $1 - 4^{-k}n$, for all $y \in \nodes_\graph{G}$:
    \begin{enumerate}
        \item if $\station{\graph{G}}[y] = 0$, then $\mtermi{\graph{G},\centers,\rp}[y] = 0$, and
        \item  if $\station{\graph{G}}[y] > 0$, then 
        \begin{equation}\label{eqn:minppr-easy-accuracy} 
            \big|\mtermi{\graph{G},\centers,\rp}[y] - \station{\graph{G}}[y]\big| 
            \le 35\tau_\graph{G}\rp \Big(1 + H(\station{\graph{G}}) - \log_2\station{\graph{G}}[y]\Big)\station{\graph{G}}[y].
        \end{equation}
        \end{enumerate}
\end{lemma}
 
\begin{proof}

    For all $y \in \nodes_\graph{G}$ with 
  $\station{\graph{G}}[y] = 0$   we say that $y$ is \emph{good}
  if $\mtermi{\graph{G},\centers,\rp}[y] = 0$. For all $y$
  with $\station{\graph{G}}[y] > 0$, we say that $y$ is good if \eqref{eqn:minppr-easy-accuracy} holds.  We will prove that each vertex $y$ is good with probability at least $1 - 4^{-k}$, splitting the proof into two cases according to $\station{\graph{G}}[y]$. The result then follows by a union bound over all $y \in \nodes_\graph{G}$.
    
    \pp{Case 1: $\boldsymbol{\station{\graph{G}}[y] = 0}$.}  
    In this case, $y$ is good if and only if for some $ X_i\in \centers$,
    $    \station{\graph{G},X_i,\rp}[y]  =0$. Since $\station{\graph{G}}[y]=0 $, 
    no vertex with positive reference rank has a path to $y$ in $\graph{G}$, so 
    for all vertices $x$ with $\station{\graph{G}}[x] \ne 0$, we have
        $\station{\graph{G},x,\rp}[y] = 0$ . Since $\dtv(\vect{p},\station{\graph{G}}) \le 1/8$, it follows that for all $i \in [k]$, $\prob(\station{\graph{G},X_i,\rp}[y] \ne 0) \le 1/8$. Since $X_1, \dots, X_k$ are independent, it follows that
    \[
        \prob(y\mbox{ is good}) = \prob(\mtermi{\graph{G},\centers,\rp}[y] = 0) \ge 1 -  8^{-k} > 1 - 4^{-k},
    \]
    as claimed.
    
    \pp{Case 2: $\boldsymbol{\station{\graph{G}}[y] \ne 0}$.} For brevity, let $\gamma(y) := \mix_\graph{G}\rp (3 - \log_2\station{\graph{G}}[y])$.
    We will show
    that, for all $y\in \nodes$,
    \begin{equation}\label{today1} \mtermi{\graph{G},\centers,\rp}[y]\geq (1 - \gamma(y)) \station{\graph{G}}[y] \end{equation}
    and that,  for all $y\in \nodes$,
    with probability at least 
    $1 - 4^{-k}$, 
    \begin{equation}\label{today2} \mtermi{\graph{G},\centers,\rp}[y] \leq \Big(1+35\tau_\graph{G}\rp \big(1 + H(\station{\graph{G}}) - \log_2\station{\graph{G}}[y]\big)\Big)\station{\graph{G}}[y].\end{equation}
    Since $\gamma(y) \leq
    35\tau_\graph{G}\rp (1 + H(\station{\graph{G}}) - \log_2\station{\graph{G}}[y])
    $,  Equation~\eqref{today1} and~\eqref{today2} imply condition (ii).

      Recall  that 
$      \mtermi{\graph{G},\centers,\rp}[y] = {\min\{\station{\graph{G},X,\rp}[y] \colon X \in \centers\} }/ {
          \Upsilon
        }$ where
        $$ 
        \Upsilon = 
    \sum_{v \in \nodes} \min\{\station{\graph{G},X,\rp}[v] \colon X \in \centers\}.
    $$
    
    First, we note that Equation~\eqref{today1} follows
    from two observations.
    \begin{itemize}
        \item  
    By   Lemma~\ref{lem:new-PR-high},
     for all $i \in [k]$ we have $\station{\graph{G}}[y] - \station{\graph{G},X_i,\rp}[y] \le \gamma(y)\station{\graph{G}}[y]$, and
    \item $\Upsilon \leq 1$.
    \end{itemize}
     
   Next, we prove that  
    for all $y\in \nodes$,
    with probability at least 
    $1 - 4^{-k}$, \eqref{today2} holds.
   
     Lemma~\ref{lem:new-minppr-upper}
    implies that if $X \sim \station{\graph{G}}$, then with probability at least $7/8$, $\station{\graph{G},X,\rp}[y] - \station{\graph{G}}[y] \le 8 \gamma(y)\station{\graph{G}}[y]$. Since $\dtv(\vect{p},\station{\graph{G}}) \le 1/8$, it follows that for all $i \in [k]$, with probability at least $3/4$, $\station{\graph{G},X_i,\rp}[y] - \station{\graph{G}}[y] \le 8\gamma(y)\station{\graph{G}}[y]$.  Thus, 
    \begin{equation}\label{today3}
        \prob\Big(
        \min\big\{\station{\graph{G},X,\rp}[y] \colon X \in \centers\big\}
          - \station{\graph{G}}[y] \le 8\gamma(y)\station{\graph{G}}[y] \Big) \ge 1-4^{-k}.
    \end{equation}
    
    To derive~\eqref{today2} from~\eqref{today3}, we next  derive a lower bound
    for~$\Upsilon$. By~\eqref{today1}, for all $i\in[k]$,
    $\station{\graph{G},X_i,\rp}[y]\geq (1-\gamma(y)) \station{\graph{G}}[y]$.
    Thus, $$\Upsilon \geq 
    \sum_{y\in \nodes\colon \station{\graph{G}}[y]>0} (1-\gamma(y)) \station{\graph{G}}[y] = 1 - \tau_\graph{G}\rp\big(3 + H(\station{\graph{G}})\big).$$
    Since $\tau_\graph{G} < 1/(2\rp(3+H(\station{\graph{G}})))$, 
    it follows that
    \[
        \Upsilon^{-1} \le 1 + 2\tau_\graph{G}\rp\big(3+H(\station{\graph{G}})\big).
    \]
    Thus, if the event of~\eqref{today3} occurs for some $y \in \nodes_\graph{G}$, then we have
    \begin{align*}
        \mtermi{\graph{G},\centers,\rp}[y] &= \min\big\{\station{\graph{G},X,\rp}[y] \colon X \in \centers\big\} \cdot \Upsilon^{-1}\\
        &\le \Big(1 + 9\gamma(y) + 2\tau_\graph{G}\rp\big(3+H(\station{\graph{G}}\big)\Big)\station{\graph{G}}[y]\\
        &= \Big(1 + \tau_\graph{G}\rp \big(33 - 9\log_2\station{\graph{G}}[y] + 2H(\station{\graph{G}}) \big) \Big)\station{\graph{G}}[y].\qedhere
    \end{align*}
    Since the event of~\eqref{today3} occurs with probability at least $1-4^{-k}$, by~\eqref{today3}, it follows that~\eqref{today2} holds for $y$ with probability at least $1-4^{-k}$ as required.
\end{proof}

{\def\thetheorem{\ref{thm:minppr-dist}}
\begin{theorem}[restated]
    \stateminpprdist
\end{theorem}
\addtocounter{theorem}{-1}
}
\begin{proof}
    Since $T(n) \le 1/(32\rp\dist\log_2 n)$, $\graph{G} \in \ergodicG_T$, and $\dist \ge 1$, Theorem~\ref{thm:dtv-low} implies that
    \[
        \dtv\big(\station{\graph{G},\R,\rp},\station{\graph{G}}) \le \mix_\graph{G}\rp (3 + H(\station{\graph{G}})\big) \le 4\rp T(n)\log_2 n \le 1/8.
    \]
    Thus by Lemma~\ref{lem:minppr-easy-accuracy}, with probability at least $1-4^{-k}n$:
    \begin{enumerate}
        \item for all $y \in \nodes_\graph{G}$ with $\station{\graph{G}}[y]=0$, we have $\mtermi{\graph{G},\centers,\rp}[y] = 0$ also; and
        \item for all $y \in \nodes_\graph{G}$ with $\station{\graph{G}}[y] \ne 0$,~\eqref{eqn:minppr-easy-accuracy} holds for $y$.
    \end{enumerate}
    
    
    Suppose this event occurs, so that (i) and (ii) hold, and
    let $y \in \nodes_\graph{G}$; 
    we will use (i) and (ii) to bound the distortion of $\mtermi{\graph{G},\centers,\rp}$ on $y$.
    We split into cases depending on $\station{\graph{G}}[y]$.
    
    \pp{Case 1: $\boldsymbol{\station{\graph{G}}[y] = 0}$.}  By (i), this implies $\mtermi{\graph{G},\centers,\rp}[y] = 0$, so $\distortion_\dist(\mtermi{\graph{G},\centers,\rp},\graph{G}) = 1$.
    
    \pp{Case 2: $\boldsymbol{0 < \station{\graph{G}}[y] < 1/n^\dist}$.} In this case we have $\cont_\dist(\mtermi{\graph{G},\centers,\rp},\graph{G},y) \le 1$, and
    \begin{equation}\label{eqn:minppr-easy-distortion}
        \str_\dist\big(\mtermi{\graph{G},\centers,\rp},\graph{G},y\big) 
        = \max\big\{1,\,n^\dist\,\mtermi{\graph{G},\centers,\rp}[y]\big\}.
    \end{equation}
    Since (ii) holds, by~\eqref{eqn:minppr-easy-accuracy} we have
    \begin{equation}\label{today4}
        n^\dist\,\mtermi{\graph{G},\centers,\rp}[y] \le n^\dist\Big(1+35\tau_\graph{G}\rp \big(1 + H(\station{\graph{G}}) - \log_2\station{\graph{G}}[y]\big)\Big)\station{\graph{G}}[y].
    \end{equation}
    The function $x\mapsto -(\log_2 x)x$ is increasing over $x \in (0,1/3]$, 
    and $\station{\graph{G}}[y] \leq 1/n^{\delta} \leq 1/3$,    so $-(\log_2 \station{\graph{G}}[y])\station{\graph{G}[y]} \le (\delta\log_2 n)/n^\delta$. We also have $H(\station{\graph{G}}) \le \log_2 n$ by a standard bound on Shannon entropy.
    It follows from~\eqref{today4} that 
    \[
          n^\dist\,\mtermi{\graph{G},\centers,\rp}[y] \le 1 + 35\tau_\graph{G}\rp\big(1 + \log_2 n + \delta\log_2 n \big) \le 1 + 105\tau_\graph{G}\rp\delta\log_2 n.
    \]
    Since $\graph{G} \in \mixclass{T}$, it follows by~\eqref{eqn:minppr-easy-distortion} that $\distortion_\dist(\station{\graph{G},\centers,\rp},\graph{G},y) \le 1 + 105\rp \dist T(n)\log_2 n$.
    
    \pp{Case 3: $\boldsymbol{\station{\graph{G}}[y] \ge 1/n^\dist}$.} 
    In this case, for brevity, let
    \[
        \Gamma = 35\tau_\graph{G}\rp\big(1 + H(\station{\graph{G}}) - \log_2\station{\graph{G}}[y]\big).
    \]
    Following Case 2, observe that $\Gamma \le 105\tau_\graph{G}\rp\delta\log_2 n$. Since $\graph{G} \in \ergodicG_T$,
    it follows that $\Gamma \le 105 \rp\delta T(n)\log_2 n \le 1/2$.
    
    Since (ii) holds, \eqref{eqn:minppr-easy-accuracy} holds for $y$, so
    \[
        \cont_\dist\big(\mtermi{\graph{G},\centers,\rp},\graph{G},y\big) \le \frac{\station{\graph{G}}[y]}{\mtermi{\graph{G},\centers,\rp}[y]} \le \frac{1}{1 - \Gamma}.
    \]
    Since $\Gamma \le 1/2$, it follows that
    \[
        \cont_\dist\big(\mtermi{\graph{G},\centers,\rp},\graph{G},y\big) \le 1+2\Gamma \le 1 + 210\rp\delta T(n)\log_2 n.
    \]
    Moreover, we have
    \[
        \str_\dist\big(\mtermi{\graph{G},\centers,\rp},\graph{G},y\big) = \max\Big\{\frac{1/n^\delta}{\station{\graph{G}}[y]},\,\frac{\mtermi{\graph{G},\centers,\rp}[y]}{\station{\graph{G}}[y]}\Big\} \le \max\Big\{1,\,\frac{\mtermi{\graph{G},\centers,\rp}[y]}{\station{\graph{G}}[y]}\Big\},
    \]
    where
    \[
        \frac{\mtermi{\graph{G},\centers,\rp}[y]}{\station{\graph{G}}[y]} \le 1 + \Gamma < 1+210\rp T(n)\log_2 n.
    \]
    The result therefore follows.
\end{proof}

We now prove a version of Lemma~\ref{lem:minppr-easy-accuracy} which gives error bounds in terms of $\emix_\graph{G}(\station{\graph{G}})$ rather than in terms of $\tau_\graph{G}$. To ensure a reasonable lower bound, we will need to discard some of our centers; we do this in a graph-agnostic way, to facilitate turning the lemma into a ranking algorithm. 


\begin{lemma}\label{lem:minppr-hard-accuracy}
    Let $\delta > 1$.  
    Let $\graph{G}=(\nodes,\edges)$ be an $n$-vertex graph in $\ergodicG$ 
    where $n\geq 2$ is sufficiently large that $n^{\delta-1}\geq 32$. 
    Let $\rp \in (0,1)$ be sufficiently small that
    $ 8 \rp\,\emix_\graph{G}(\station{\graph{G}}) \leq 1/(20  \delta  \log_2(n))$.
    Let $\gamma$ be any real number satisfying
    $ 8 \rp\,\emix_\graph{G}(\station{\graph{G}}) \leq \gamma \leq 1/(20  \delta  \log_2(n))$.

    Let $\vect{p}$ be a probability distribution on $\nodes$ with $\dtv(\vect{p},\station{\graph{G}}) \le 1/8$. Let $X_1, \dots, X_{2k+1} \sim \vect{p}$ be independent and identically distributed. For each $y \in \nodes$, let $M(y)$ be the median of $\{\station{\graph{G},X_i,\rp}[y] \colon i \in [2k+1]\}$.   For each $i \in [2k+1]$, let 
    \[
        \xi_i := \max\Big\{\frac{M(y) - \station{\graph{G},X_i,\rp}[y]}{M(y)} \colon y \in \nodes,\,M(y) \ge 1/(2 n^{\delta} )\Big\}.
    \]
    Let $f\colon [2k+1]\to[2k+1]$ be an arbitrary permutation such that $\xi_{f(1)} \le \dots \le \xi_{f(2k+1)}$, and let $\centers = \{X_{f(i)} \colon i \in [k+1]\}$.  
    Then with probability at least $1-(n+1)e^{-k/6}$, for all $y \in \nodes$,
    \begin{equation}\label{eqn:foo-1}
     \mtermi{\graph{G},\centers,\rp}[y] 
      \le \begin{cases}
           \big(1+2\gamma 
        (3-\log_2\station{\graph{G}}[y]  ) 
            + 12\delta \gamma\log_2 n + 2n^{1-\delta}\big)
           \station{\graph{G}}[y]
           & \mbox{ if }\station{\graph{G}}[y] \ne 0,\\
            0 & \mbox{ otherwise.}
        \end{cases}     
    \end{equation}
    Moreover, for all $y \in \nodes$ with $\station{\graph{G}}[y] \ge n^{-\delta}$,
    \begin{equation}\label{eqn:foo-2} 
      \mtermi{\graph{G},\centers,\rp}[y]    \ge (1 - 6\delta \gamma\log_2 n)\station{\graph{G}}[y].
    \end{equation}
\end{lemma}

\begin{proof}

  We use the following definitions throughout the proof.
For all 
$y \in \nodes$ with $\station{\graph{G}}[y] \ne 0$, we let  $L(y) = 3-\log_2\station{\graph{G}}[y]$.
Let $\mathbf{A}$ be the set of all indices $i \in [2k+1]$ such that for some $y \in \nodes$ with $\station{\graph{G}}[y] \ge 1/(4 n^\delta)$, $\station{\graph{G},X_i,\rp}[y] < (1-\gamma L(y))\station{\graph{G}}[y]$. 
Finally, for all $y \in \nodes$, let 
    \[
        \mathbf{B}_y = \begin{cases}
            \big\{i \in [2k+1] \colon \station{\graph{G},X_i,\rp}[y] > (1+\gamma L(y))\station{\graph{G}}[y]\big\} & \mbox{ if }\station{\graph{G}}[y] \ne 0,\\
            \big\{i \in [2k+1] \colon \station{\graph{G},X_i,\rp}[y] \ne 0\big\} & \mbox{ otherwise.}
        \end{cases}
    \]
    We will first bound $|\mathbf{A}|$ and each $|\mathbf{B}_y|$ above with high probability, then use these bounds to prove  the lemma.

    {\bf Claim 1: } 
    $\prob\big(|\mathbf{A}| \le k \wedge (\forall y \in \nodes, |\mathbf{B}_y| \le k )\big) \geq 1 - (n+1)e^{-k/6}$.
 
    {\bf Proof of Claim 1:} We first bound $|\mathbf{A}|$. Since $\gamma \ge 8\rp\,\emix_{\graph{G}}(\station{\graph{G}})$, by Lemma~\ref{lem:new-minppr-lower}, if $X\sim\station{\graph{G}}$ then \[
        \prob\Big(\station{\graph{G},X,\rp}[y] < (1-\gamma L(y))\station{\graph{G}}[y] \mbox{ for some }y \in \nodes\mbox{ with } \station{\graph{G}}[y] > 0\Big) \le 1/8.
    \]
    Since $\dtv(\station{\graph{G}}, \vect{p}) \le 1/8$, it follows that
    for any $i\in [2k+1]$, we have
    $\prob(i \in \mathbf{A}) \le 1/4$, so $|\mathbf{A}|$ is a binomial variable with mean at most $(2k+1)/4$. Thus by a Chernoff bound (Lemma~\ref{lem:new-chernoff} applied with $\eta = 1$),
    it follows that
    \begin{equation}\label{eqn:minppr-accurate-hard}
        \prob(|\mathbf{A}| \ge k+1) \le \prob\Big(|\mathbf{A}| \ge 2\cdot \frac{2k+1}{4}\Big) \le e^{-(2k+1)/12} \le e^{-k/6}.
    \end{equation}
    
    We now bound each $|\mathbf{B}_y|$. 
    \begin{itemize}
        \item 
 First, suppose $y \in \nodes$ with $\station{\graph{G}}[y] = 0$. 
 
  As in the proof of Lemma~\ref{lem:minppr-easy-accuracy}, no vertex with positive reference rank has a path to $y$ in $\graph{G}$, so if $X \sim \station{\graph{G}}$ then $\prob(\station{\graph{G},X,\rp}[y] \ne 0) = 0$. Since $\dtv(\vect{p},\station{\graph{G}}) \le 1/8$, 
  for each $i\in [2k+1]$
  it follows that  $\prob(\station{\graph{G},X_i,\rp}[y] \ne 0) \le 1/8$. Thus $|\mathbf{B_y}|$ is a binomial variable with mean at most $(2k+1)/8$, so by Lemma~\ref{lem:new-chernoff} we have
    \begin{equation}\label{eqn:minppr-accurate-hard-2}
        \prob(|\mathbf{B}_y| \ge k+1) \le e^{-k/6} \mbox{ whenever }\station{\graph{G}}[y] = 0.
    \end{equation}
    \item 
    Next suppose $y \in \nodes$ with $\station{\graph{G}}[y]\ne 0$. 
    
    Since $\gamma \ge 8\rp\,\emix_\graph{G}(\station{\graph{G}})$, by Lemma~\ref{lem:new-minppr-upper}, if $X \sim \station{\graph{G}}$ then
    \[
        \prob\Big(\station{\graph{G},X,\rp}[y] > (1+\gamma L(y))\station{\graph{G}}[y] \Big) \le 1/8.
    \]
    Since $\dtv(\station{\graph{G}}, \vect{p}) \le 1/8$, it follows that $\prob(i \in \mathbf{B}_y) \le 1/4$. Thus once again by Lemma~\ref{lem:new-chernoff}, we have
    \begin{equation}\label{eqn:minppr-accurate-hard-3}
        \prob(|\mathbf{B}_y| \ge k+1) \le e^{-k/6} \mbox{ whenever }\station{\graph{G}}[y] \ne 0.
    \end{equation}
    \end{itemize}
    Combining~\eqref{eqn:minppr-accurate-hard}--\eqref{eqn:minppr-accurate-hard-3} with a union bound,  the claim follows.
    \hfil {\bf (End of Proof of Claim 1.)}
    
     The lemma follows from Claim~1, together with the following claim.
    
    {\bf Claim 2:} If 
    $|\mathbf{A}| \le k$ and, for all $y \in \nodes$, $|\mathbf{B}_y| \le k$, 
    then:
    \begin{itemize}
    \item for all $y\in \nodes$, inequality~\eqref{eqn:foo-1} holds, and
    \item for all $y \in \nodes$ with $\station{\graph{G}}[y] \ge n^{-\delta}$,
     inequality~\eqref{eqn:foo-2} holds.
     \end{itemize}
    
    {\bf Proof of Claim 2:}
    From now on we will assume  that $|\mathbf{A}| \le k$ and, for all $y \in \nodes$, that $|\mathbf{B}_y| \le k$.
    
    Equations~\eqref{eqn:foo-1} and~\eqref{eqn:foo-2} give upper and lower bounds on
  $ \mtermi{\graph{G},\centers,\rp}[y]$ for certain $y\in \nodes$. 
      Recall that $ \mtermi{\graph{G},\centers,\rp}[y] = {\min\{\station{\graph{G},X,\rp}[y] \colon X \in \centers\} }/ {
          \Upsilon
        } $
        where $$\Upsilon = 
    \sum_{v \in \nodes} \min\big\{\station{\graph{G},X,\rp}[v] \colon X \in \centers\big\}.
    $$

    We first observe that for all $y \in \nodes$ and all $(k+1)$-element subsets $\mathbf{I}$ of $[2k+1]$, since $|\mathbf{B}_y| \le k$,
    \begin{equation}\label{eqn:minppr-accurate-hard-4}
        \min\big\{\station{\graph{G},X_i,\rp}[y] \colon i \in \mathbf{I}\big\} \le \begin{cases}
            (1+\gamma L(y))\station{\graph{G}}[y] & \mbox{ if }\station{\graph{G}}[y] \ne 0,\\
            0 & \mbox{otherwise}.
        \end{cases}
    \end{equation}

     It will be useful to upper-bound the right-hand-side of~\eqref{eqn:minppr-accurate-hard-4} 
     as follows.
    Suppose $0 < \station{\graph{G}}[y] \leq r \leq 1/3$.
    In this case,
    \begin{equation}\label{eq:calc}
    (1+ \gamma L(y)) \station{\graph{G}}[y] \leq (1 + \gamma (3 - \log_2 r)) r.\end{equation}
    To see this, note that the function $f(x) = x \log_2(1/x)$ is increasing for $x\in (0,1/3]$,
    so
    \begin{align*}
    (1+ \gamma (3- \log_2 \station{\graph{G}}[y])) \station{\graph{G}}[y] &=
    1 + 3 \gamma \station{\graph{G}}[y] + \gamma \station{\graph{G}}[y] \log_2(1/\station{\graph{G}}[y])\\
    &\leq 1 + 3 \gamma r  + \gamma r \log_2(1/r)\\
    &= 
    \big(1 + \gamma (3 - \log_2 r)\big) r.
    \end{align*}

    The  proof of the claim will proceed as follows. In Step~1, we  prove an upper bound on $\xi_i$ for all $i \notin \mathbf{A}$, and hence (as we will see)
    for all $i$ such that $X_i \in \centers$. In Step~2, 
    we will   turn this into a lower bound on $\min\{\station{\graph{G},X,\rp}[y]\colon X \in \centers\}$ whenever $\station{\graph{G}}[y] \ge 1/n^{\delta}$. This will suffice in Step~3 to bound the normalizing factor~$\Upsilon$ below. In Step~4, we use this, together with~\eqref{eqn:minppr-accurate-hard-4} and the lower bound of Step~2, to prove~\eqref{eqn:foo-1} and~\eqref{eqn:foo-2}.

    {\bf Step 1:}
Consider any 
 $v \in \nodes$ with $M(v) \ge 1/(2 n^{\delta})$.  
 Each time we apply~\eqref{eqn:minppr-accurate-hard-4} in this step
 we will take $\mathbf{I}$ to be a set of $k+1$ indices $i$ with $\station{\graph{G},X_i,\rp}[v]$ as large as possible.

    We first prove $\station{\graph{G}}[v] >0$. For contradiction, suppose 
    $\station{\graph{G}}[v] =0$.
    By~\eqref{eqn:minppr-accurate-hard-4}  we have $M(v)=0$, contradicting
    our choice of $v$.

    Given that 
    $\station{\graph{G}}[v] >0$ we now prove $\station{\graph{G}}[v] \ge 1/(4n^{\delta})$. For contradiction, suppose $\station{\graph{G}}[v] < 1/(4 n^{\delta})$. By~\eqref{eqn:minppr-accurate-hard-4}  we have $M(v) \le (1 + \gamma L(y))\station{\graph{G}}[v]$. 
     Using~\eqref{eq:calc} with $r=1/(4 n^{\delta})$, we have   
     $$M(v) \leq   \big(1 + \gamma(3 + \log_2(4 n^\delta))\big) \frac{1}{4 n^\delta}
     = \big(1 + \gamma(5 + \delta \log_2( n ))\big) \frac{1}{4 n^\delta}
     .$$ Since (from the statement) 
     $\delta \log_2 n \geq 5$,  this is at most 
     $(1+ 2 \gamma \delta \log_2(n)) \frac{1}{4n^{\delta}}$.  Using the upper bound on $\gamma$ from the statement of the lemma, this is less than $1/(2 n^{\delta})$.  This contradicts our choice of $v$, so we must have $\station{\graph{G}}[v] \ge 1/(4n^{\delta})$ as claimed. 
    
    Now consider any $i\in [2k+1]\setminus \mathbf{A}$.  It follows from the definition of $\mathbf{A}$ that
    \begin{equation}\label{eqn:foo-5}
        \station{\graph{G},X_i,\rp}[v] \ge (1-\gamma L(v))\station{\graph{G}}[v].
    \end{equation} 
    
     Once again, by~\eqref{eqn:minppr-accurate-hard-4} 
     we have
      $M(v) \le (1+\gamma L(v))\station{\graph{G}}[v]$, so 
      \[
        \station{\graph{G}}[v] \ge (1+\gamma L(v))^{-1}M(v) \ge (1-\gamma L(v))M(v).
      \]
     
      It follows from~\eqref{eqn:foo-5} that $\station{\graph{G},X_i,\rp}[v] \ge (1-2\gamma L(v))M(v)$. Since $\station{\graph{G}}[v] \ge 1/(4n^\delta)$, it follows using the definition of $L(v)$ that
    $$
        \station{\graph{G},X_i,\rp}[v] \geq  \big(1 - 2\gamma(3 + \log_2 (4n^\delta)) \big)M(v) 
         = \big(1 - 2\gamma(5 +  \delta \log_2 n )\big)M(v).
    $$
    Since  (from the statement)
     $\delta \log_2 n \geq 5$,  this is at   least $(1-4 \gamma\delta \log_2 (n) )M(v)$.
     We have shown that for every $v\in \nodes$ with 
     $M(v) \ge 1/(2 n^{\delta})$, we have 
     $ \station{\graph{G},X_i,\rp}[v] \geq (1-4 \gamma\delta \log_2 (n) )M(v) $.

    From the definition of $\xi_i$, 
    there is some $v\in \nodes$ with 
     $M(v) \ge 1/(2 n^{\delta})$ such that
    $\xi_i = 1 -  \station{\graph{G},X_i,\rp}[v]/M(v)$.
   Thus $\xi_i \le  4\gamma\delta \log_2 n$ for all $i \notin \mathbf{A}$. Since $|\mathbf{A}| \le k$ and $|\centers| = 2k+1$, it follows that $\xi_{f(1)}, \dots, \xi_{f(k+1)} \le 4\gamma\delta \log_2n$.
    
    {\bf Step 2:} 
    Consider any $y\in \nodes$ with $\station{\graph{G}}[y] \ge 1/n^\delta$.
    By the definition of~$\mathbf{A}$, since $|\mathbf{A}|\leq k$,
    we have
    $M(y) \geq (1 - \gamma L(y))\station{\graph{G}}[y] $.
Since $\delta \log_2 n \geq 3$,    
    $$
L(y) = 3 + \log_2(1/\station{\graph{G}}[y] )
\leq 3 + \delta \log_2 n \leq 2 \delta \log_2 n,
 $$
so we conclude that
$M(y) \geq (1 - \gamma L(y))\station{\graph{G}}[y]\geq (1-2\gamma\delta \log_2 n)\station{\graph{G}}[y] $.
Since the upper bound on $\gamma$ in the statement guarantees that $\gamma \delta \log_2 n \leq 1/4$,
we have 
    \begin{equation}\label{eqn:foo-med-lower}
        M(y) \ge  (1-2\gamma\delta\log_2 n)\station{\graph{G}}[y] \ge \station{\graph{G}}[y]/2 \ge 1/(2n^\delta).
    \end{equation}

   We proved in Step~1 that   $\xi_{f(1)}, \dots, \xi_{f(k+1)} \leq 4\gamma\delta \log_2 n$.
   Thus, by the definition of $\centers$ in the statement of the lemma, 
   for each $X_i \in \centers$, $\xi_i \leq 
   4\gamma\delta \log_2 n$.
   Since \eqref{eqn:foo-med-lower} guarantees that $M(y) \geq  1/(2n^\delta)$,
   the definition of $\xi_i$
   ensures that 
   $\xi_i \geq 1-
      \station{\graph{G},X_i,\rp}[y]/{M(y)}$ so 
    $$ \min \{\station{\graph{G},X,\rp}[y] \colon X \in \centers \} \geq (1-4\gamma \delta \log_2 n) M(y).$$
    Using the first inequality in~\eqref{eqn:foo-med-lower}, 
we have
    \begin{equation}\label{eqn:foo-3}
        \min\big\{\station{\graph{G},X,\rp}[y] \colon X \in \centers\big\} \ge (1-6\gamma\delta\log_2 n)\station{\graph{G}}[y].
    \end{equation}
    Note that we have proved \eqref{eqn:foo-3} for all 
  $y\in \nodes$ with $\station{\graph{G}}[y] \geq 1/n^\delta$.  
    
    {\bf Step 3:} We next bound the normalizing factor~$\Upsilon$. By~\eqref{eqn:foo-3}, we have
    \begin{align*}
        \Upsilon = \sum_{v \in \nodes} \min\big\{\station{\graph{G},X,\rp}[v] \colon X \in \centers\big\} 
        &\geq \sum_{\substack{v \in \nodes\\\station{\graph{G}}[v] \geq 1/n^\delta}}(1-6\gamma\delta\log_2 n)\station{\graph{G}}[v]\\
        &\geq 1 - 6\gamma\delta\log_2 n - n(1/n^\delta)
        =1 - 6\gamma\delta\log_2 n - n^{1-\delta}.
    \end{align*}
The upper bound on $\gamma$ in the statement of the lemma guarantees that
$6\gamma\delta\log_2 n\leq 6/20$. Since $n^{\delta-1}\geq 5$, 
$n^{1-\delta} \leq 4/20$ so 
the sum of these is at most~$1/2$. It follows that 
    \begin{equation}\label{eqn:foo-4}
         \frac{1}{\Upsilon} \leq 1 + 12\gamma\delta\log_2 n +2 n^{1-\delta}. 
    \end{equation}
    
    {\bf Step 4:} We are now ready to prove~\eqref{eqn:foo-1} and~\eqref{eqn:foo-2}.
    For all $y \in \nodes$ with $\station{\graph{G}}[y] \geq n^{-\delta}$, since $ \mtermi{\graph{G},\centers,\rp}[y] = {\min\{\station{\graph{G},X,\rp}[y] \colon X \in \centers\} }/\Upsilon$
    and
    $\Upsilon \leq 1$,~\eqref{eqn:foo-3} implies that~\eqref{eqn:foo-2} holds. By~\eqref{eqn:minppr-accurate-hard-4} applied with $\mathbf{I} = \{f(1),\dots,f(k+1)\}$, for all $y \in \nodes$ with $\station{\graph{G}}[y] = 0$, we have $\mtermi{\station{\graph{G}},\centers,\rp}[y] = 0$ (as required by~\eqref{eqn:foo-1}). Finally, again by~\eqref{eqn:minppr-accurate-hard-4}, for all $y \in \nodes$ with $\station{\graph{G}}[y] \ne 0$, we have
    \[
        \min\big\{\station{\graph{G},X,\rp}[y] \colon X \in \centers\big\} \leq 
        (1 + \gamma L(y))\station{\graph{G}}[y].
    \]

    The upper bound on $\gamma$ in the statement of the lemma ensures that
    $12 \gamma \delta \log_2 n \leq 3/5$. 
    Since $n^{\delta-1} \geq 5$, 
    $2 n^{1-\delta} \leq 2/5$. Thus, their sum (in the right-hand side of~\eqref{eqn:foo-4}) is at most~$1$. It follows by~\eqref{eqn:foo-4} that
    \[
         \mtermi{\station{\graph{G}},\centers,\rp} \leq 
         \big(1+2\gamma L(y) 
            + 12\delta \gamma\log_2 n + 2n^{1-\delta}\big)
           \station{\graph{G}}[y].
    \]
    Hence~\eqref{eqn:foo-1} follows.
      \hfil {\bf (End of Proof of Claim 2.)}\end{proof}

    \newcommand{\hardmin}{\textnormal{T-Min-PPR}}
    We now turn Lemma~\ref{lem:minppr-hard-accuracy} into a ranking algorithm $\hardmin_{\gamma,\delta,k,\rp}$ as follows. The parameter~$k$ is a positive integer and the other parameters are real numbers satisfying
    $\gamma,\rp \in (0,1)$ and $\delta>1$. Given a graph $\graph{G}$ and a trusted set $\trust \subseteq \nodes_\graph{G}$, $\hardmin_{\gamma,\delta,k,\rp}(\graph{G},\trust)$ chooses a set $\centers \subseteq \trust$ of size $\min\{2k-1,|\trust|\}$. Then the algorithm calculates, for each $y \in \nodes$, the median $M(y)$ of $\{\station{\graph{G},c,\rp} \colon c \in \centers\}$, and the observed divergences
    \[
        \xi_c := \max\Big\{\frac{M(y) - \station{\graph{G},c,\rp}[y]}{M(y)} \colon y \in \nodes,\,M(y) \ge \frac{1}{2{n^{\delta}}} \Big\}.
    \]
    The algorithm then forms a set $\centers' \subseteq \centers$ by discarding the $k-1$ vertices in $\centers$ with the highest values of $\xi_c$, then taking a   coherent subset that is as large as possible. Finally, the algorithm outputs $\mtermi{\graph{G},\centers',\rp}$. 
    
    Essentially, $\hardmin_{\gamma,\delta,k,\rp}$ is similar to T-Min-PPR$_{k,\rp}$, except that, rather than choosing  $k$ centers arbitrarily from $\trust$, the algorithm chooses  them according to which of their PPRs agrees most closely with the median PPR. As the following theorem shows, when $\graph{G} \in \ergodicG$, $n^{-\delta}$ acts as a significance threshold, $\gamma$ acts as an accuracy parameter for vertices with reference rank above this threshold, and the algorithm gives good results when $\rp$ is small relative to $n^{-\delta}$, $\gamma$, and $\tau_\graph{G}$.
    
    \begin{theorem}\label{thm:minppr-dist-hard}
 Let $k$ be a positive integer.
   Let $\delta > 1$.  
    Let $\graph{G}=(\nodes,\edges)$ be an $n$-vertex graph in $\ergodicG$ 
    where $n\geq 2$ is sufficiently large that $n^{\delta-1}\geq 32$. 
    Let $\rp \in (0,1)$ be sufficiently small that
    $ 8 \rp\,\emix_\graph{G}(\station{\graph{G}}) \leq 1/(20  \delta  \log_2(n))$.
    Let $\gamma$ be any real number satisfying
    $ 8 \rp\,\emix_\graph{G}(\station{\graph{G}}) \leq \gamma \leq 1/(20( 10+ \delta  \log_2(n)))$.
      Let $\vect{p}$ be a probability distribution on $\nodes$ with $\dtv(\vect{p},\station{\graph{G}}) \le 1/8$. Let $X_1, \dots, X_{2k-1} \sim \vect{p}$ be independent and identically distributed.
     Let $\trust = \{X_1, \dots, X_{2k-1}\}$.  
  Then with probability at least $1-(n+1)e^{-(k-1)/6}$,
        \begin{equation}\label{eqn:minppr-hard}
            \distortion_\dist\big(\hardmin_{\gamma,\delta,k,\rp}(\graph{G},\trust),\graph{G}\big) \le 1+40\gamma\delta\log_2 n + 2n^{1-\delta}.
        \end{equation}
    \end{theorem}
    
    Theorem~\ref{thm:minppr-dist-hard} says that when $\graph{G}$ is ergodic, and the parameters are chosen appropriately, then $\hardmin_{\gamma,\delta,k,\rp}$ performs essentially at least as well as the simpler algorithm T-Min-PPR$_{k,\rp}$. The key difference is that T-Min-PPR$_{k,\rp}$ requires an upper bound on the worst-case mixing time $\tau_\graph{G}$, while $\hardmin_{\gamma,\delta,k,\rp}$ requires an analogous upper bound on the (potentially much smaller) average-case mixing time $\emix_\graph{G}(\station{\graph{G}})$.  
    
    \begin{proof}
        Recall from the definition that in computing a ranking function, $\hardmin_{\gamma,\delta,k,\rp}(\graph{G},\trust)$ first chooses a subset of $\trust$ of size $\min\{2k-1,|\trust|\}$; since $|\trust|=2k-1$, this must be $\trust$ itself. It then chooses a subset $\centers \subseteq \trust$ by discarding $k-1$ vertices as in the statement of Lemma~\ref{lem:minppr-hard-accuracy}. (Note that in this proof, we will take the $k$ of Lemma~\ref{lem:minppr-hard-accuracy} to be our present $k-1$.) Since $\graph{G} \in \ergodicG$, any non-empty subset of $\nodes$ is coherent; thus $\centers$ is coherent, so we have $\hardmin_{\gamma,\delta,k,\rp}(\graph{G},\trust) = \mtermi{\graph{G},\centers,\rp}$.
        
        By Lemma~\ref{lem:minppr-hard-accuracy}, with probability at least $1 - (n+1)e^{-(k-1)/6}$,~\eqref{eqn:foo-1} holds for all $y \in \nodes$ and~\eqref{eqn:foo-2} holds for all $y \in \nodes$ with $\station{\graph{G}}[y] \ge n^{-\delta}$. 
        Suppose this event occurs; then we will show that~\eqref{eqn:minppr-hard} holds. To bound the distortion at each vertex $y \in \nodes$, we split into cases depending on $\station{\graph{G}}[y]$.

        \medskip\noindent\textbf{Case 1: $\boldsymbol{\station{\graph{G}}[y] = 0}$.} By~\eqref{eqn:foo-1}, this implies $\mtermi{\graph{G},\centers,\rp}[y] = 0$. Thus $\mtermi{\graph{G},\centers,\rp}$, and hence $\hardmin_{\gamma,\delta,k,\rp}(\graph{G},\trust)$, has distortion exactly $1$ at $y$.
        
        \medskip\noindent\textbf{Case 2: $\boldsymbol{0 < \station{\graph{G}}[y] < 1/n^\delta}$.} In this case, $\mtermi{\graph{G},\centers,\rp}$ has contraction at most $1$ at $y$, and
        \[
            \str_\dist\big(\mtermi{\graph{G},\centers,\rp},\graph{G},y\big) = \max\big\{1,n^\delta \mtermi{\graph{G},\centers,\rp}[y]\big\}.
        \]
        Since~\eqref{eqn:foo-1} holds for $y$, we have 
        \[
            n^\delta \mtermi{\graph{G},\centers,\rp}[y] \le n^\delta\big(1 + 2\gamma(3-\log_2\station{\graph{G}}[y])+12\gamma\delta \log_2 n  + 2n^{1-\delta} \big)\station{\graph{G}}.
        \]
        The function $x\mapsto -(\log_2x)x$ is increasing over $x \in (0,1/3]$, and $\station{\graph{G}}[y] \le 1/n^\delta \le 1/3$, so $-(\log_2\station{\graph{G}}[y])\station{\graph{G}}[y] \le (\delta\log_2 n)/n^\delta$.  It follows that
        \begin{align*}
            n^\delta \mtermi{\graph{G},\centers,\rp}[y]
            &\le 1+2\gamma(3+\delta\log_2 n) + 12\gamma\delta\log_2 n + 2n^{1-\delta}
            \le 1+20\gamma\delta\log_2 n + 2n^{1-\delta}.
        \end{align*}
        It follows that $\distortion_\dist(\mtermi{\graph{G},\centers,\rp},\graph{G},y) \le 1 + 40\gamma\delta\log_2 n + 2n^{1-\delta}$, as required.
        
        \medskip\noindent\textbf{Case 3: $\boldsymbol{\station{\graph{G}}[y] \ge 1/n^\delta}$.} Since~\eqref{eqn:foo-2} holds for all $y \in \nodes$ with $\station{\graph{G}}[y] \ge n^{-\delta}$, for all such $y$ we have
        \[
            \cont_\dist\big(\mtermi{\graph{G},\centers,\rp},\graph{G},y\big) \le \frac{\station{\graph{G}}[y]}{\mtermi{\graph{G},\centers,\rp}[y]} \le \frac{1}{1-6\delta\gamma\log_2 n}.
        \]
        Since $\gamma \le 1/(20(10+\delta\log_2 n))$, we have $6\delta\gamma\log_2 n \le 1/2$; hence
        \begin{equation}\label{eqn:final-dist}
            \cont_\dist\big(\mtermi{\graph{G},\centers,\rp},\graph{G},y\big) \le 1 + 12\delta\gamma\log_2 n < 1 + 40\gamma\delta\log_2 n + 2n^{1-\delta}.
        \end{equation}
        Moreover, we have
        \[
            \str_\dist\big(\mtermi{\graph{G},\centers,\rp},\graph{G},y\big) = \max\Big\{\frac{1/n^\delta}{\station{\graph{G}}[y]},\,\frac{\mtermi{\graph{G},\centers,\rp}[y]}{\station{\graph{G}}[y]}\Big\} \le \max\Big\{1,\,\frac{\mtermi{\graph{G},\centers,\rp}[y]}{\station{\graph{G}}[y]}\Big\},
        \]
        where~\eqref{eqn:foo-1} implies that
        \[
            \frac{\mtermi{\graph{G},\centers,\rp}[y]}{\station{\graph{G}}[y]} \le 1+40\gamma\delta\log_2 n + 2n^{1-\delta}.
        \]
        The result therefore follows from~\eqref{eqn:final-dist}.
    \end{proof}

%% file: new-spam-resistance.tex

\section{Spam Resistance of \minppr}
\label{sec:termispam}

In this section, we bound the spam resistance of $\mtermialgo{k}{\rp}$, proving Theorem~\ref{thm:min-spam-resistant}.
We first prove a technical lemma which bounds the effect a spammer can have on PPR.

\begin{lemma}\label{lem:PPR}
    Let $\graph{G} = (\nodes,\edges_\graph{G})$ be a graph. Let $\purchase \subseteq \nodes$, and let $\graph{H} = (\nodes \cup \spam, \edges_\graph{H}) \in \graph{G}_\purchase$ (where $\spam$ and $\nodes$ are disjoint). Let $\R$ be a reset vector on $\graph{G}$ with $\R[v] = 0$ for all $v \in \purchase$, and let $\R'$ be the corresponding reset vector on $\graph{H}$ with $\R'[v] = \R[v]$ for all $v \in \nodes$ and $\R'[v] = 0$ for all $v \in \spam$. Then for all $0 < \rp < 1$ and all $\mathbf{A} \subseteq \nodes$, we have
    \[
        \station{\graph{H},\R',\rp}[\mathbf{A} \cup \spam] \le \station{\graph{G},\R,\rp}[\mathbf{A} \setminus \purchase] + \rp^{-1}\station{\graph{G},\R,\rp}[\purchase]
    \]
\end{lemma}

Remark: By taking $\mathbf{A} = \purchase$, Lemma~\ref{lem:PPR} implies that $\mTPPR$ is $\rp$-spam resistant (see the proof of Lemma~\ref{lem:ppr-simple} below). In fact, the same holds for any ranking algorithm which carries out some form of PageRank that only resets to trusted vertices. However, Lemma~\ref{lem:PPR} does not immediately imply $\rp$-spam resistance for $\mtermialgo{k}{\rp}$, since
the reset vector on $\graph{H}$ does not in general match the reset vector on $\graph{G}$ as required by the lemma. Nevertheless, we will use the full strength of the lemma (with a more subtle choice of~$\mathbf{A}$) to demonstrate spam resistance
of $\mtermialgo{k}{\rp}$ in the proof of Theorem~\ref{thm:min-spam-resistant} below.  

\begin{proof}
    Let $(X_i)_{i \ge 0}$ be a uniform random walk on $\graph{G}$ with initial state drawn from $\R$, and let $(Y_i)_{i \ge 0}$ be a uniform random walk on $\graph{H}$ with initial state drawn from $\R'$. By Lemma~\ref{lem:stat-dist} applied to $\graph{H}$ and $\R'$, we have
    \begin{equation}\label{eqn:PPR-proof-0}
        \station{\graph{H},\R',\rp}[\mathbf{A} \cup \spam] = \rp\sum_{i=0}^\infty (1-\rp)^i\prob(Y_i \in \mathbf{A} \cup \spam).
    \end{equation}
    If $Y_i \in \mathbf{A} \cup \spam$, then either $Y$ passed through $\purchase$ at some time in $[0,i]$ or it did not, and in the latter case we must have $Y_i \in \mathbf{A} \setminus \purchase$. For all $j \ge 0$, let $\calE_j$ be the event that $Y_0, \dots, Y_{j-1} \notin \purchase$. Then by~\eqref{eqn:PPR-proof-0} and a union bound, we have
    \[
        \station{\graph{H},\R',\rp}[\mathbf{A} \cup \spam] \le \rp\sum_{i=0}^\infty (1-\rp)^i \Big(\sum_{j=0}^i \prob\big(Y_j \in \purchase \mbox{ and }\calE_j\big) + \prob\big(Y_i \in \mathbf{A} \setminus \purchase \mbox{ and }\calE_i\big)\Big)
    \]
    Before hitting $\purchase$, $Y$ behaves exactly like $X$; more formally, the two chains can be coupled until this stopping time. 
    Thus
    \begin{align}\nonumber
        \station{\graph{H},\R',\rp}[\mathbf{A} \cup \spam] &\le \rp\sum_{i=0}^\infty (1-\rp)^i \Big(\sum_{j=0}^i \prob\big(X_j \in \purchase \mbox{ and }\calE_j\big) + \prob\big(X_i \in \mathbf{A} \setminus \purchase \mbox{ and }\calE_i\big)\Big)\\\label{eqn:PPR-proof-1}
        &\le \rp\sum_{i=0}^\infty (1-\rp)^i \Big(\sum_{j=0}^i \prob\big(X_j \in \purchase\big) + \prob\big(X_i \in \mathbf{A} \setminus \purchase\big)\Big)
    \end{align}
    
    We now simplify each part of the right-hand side of~\eqref{eqn:PPR-proof-1}. By Lemma~\ref{lem:stat-dist} applied to $\graph{G}$ and $\R$, we have
    \begin{equation}\label{eqn:PPR-proof-2}
        \rp\sum_{i=0}^\infty (1-\rp)^i\prob(X_i \in \mathbf{A} \setminus \purchase) = \station{\graph{G},\R,\rp}[\mathbf{A} \setminus \purchase].
    \end{equation}
    Moreover, by reordering the summation, 
    we have
    \begin{equation*}
        \rp\sum_{i=0}^\infty (1-\rp)^i \sum_{j=0}^i \prob\big(X_j \in \purchase\big) = \rp\sum_{j=0}^\infty \prob(X_j \in \purchase)\sum_{i=j}^\infty (1-\rp)^i = \sum_{j=0}^\infty \prob(X_j \in \purchase)(1-\rp)^j
    \end{equation*}
    By Lemma~\ref{lem:stat-dist} applied to $\graph{G}$ and $\R$, it follows that
    \[
        \rp\sum_{i=0}^\infty (1-\rp)^i \sum_{j=0}^i \prob\big(X_j \in \purchase\big) = \frac{1}{\rp}\station{\graph{G},\R,\rp}[\purchase].
    \]
    The result therefore follows from~\eqref{eqn:PPR-proof-1} and~\eqref{eqn:PPR-proof-2}.
\end{proof}

We now use Lemma~\ref{lem:PPR} to prove our spam resistance results for T-PPR and T-Min-PPR.

{\def\thetheorem{\ref{lem:ppr-simple}}
\begin{lemma}[restated]
    \statepprsimple
\end{lemma}
\addtocounter{theorem}{-1}
}

\begin{proof}
    Let $\graph{G} = (\nodes,\edges_\graph{G})$ be a graph, let $\trust_\graph{G} \subseteq \nodes$ be non-empty, let $c \in \trust_\graph{G}$ be the center chosen by $\mTPPR$, and let $\rp \in (0,1)$. We define our cost function by $C[v] :=  \norm{\station{\graph{G},c,\rp}}[v]$. (Recall that we write $\norm{\scaledvect{x}} := \scaledvect{x}/||\scaledvect{x}||$; here the normalization is necessary since $C$ is only defined on $\nodes \setminus \trust_\graph{G}$.) Let $\purchase \subseteq \nodes \setminus \trust_\graph{G}$, and let $\graph{H} = (\nodes \cup \spam,\edges_\graph{H}) \in \graph{G}_\purchase$, where $\nodes$ and $\spam$ are disjoint. Then $\mTPPR(\graph{H},\trust_\graph{G}) = \station{\graph{H},c,\rp}$, 
    so by Lemma~\ref{lem:PPR} (applied with $\mathbf{A} = \purchase$ and $\R = c$) we have
    \[
        \mTPPR(\graph{H},\trust_\graph{G})[\spam \cup \purchase] \le \rp^{-1}\station{\graph{G},c,\rp}[\purchase] \le C[\purchase]/\rp,
    \]
    as required by the definition of spam resistance.
\end{proof}

{\def\thetheorem{\ref{thm:min-spam-resistant}}
\begin{theorem}[restated]
    \stateminspamresistant
\end{theorem}
\addtocounter{theorem}{-1}
}

\begin{proof}
    Let $\graph{G} = (\nodes,\edges_\graph{G}) \in \ergodicG$ have $n$ vertices and satisfy 
    $\mix_{\graph{G}} \le 1/    (3\rp(2+\log_2 n)) $, 
    and let $\trust_\graph{G} \subseteq \nodes$ be non-empty. Let $\centers$ be the subset of $\trust_\graph{G}$ chosen by $\mtermialgot{k}{\rp}$; thus $|\centers| \le k$. (Recall that $\centers$ depends only on $\trust_\graph{G}$, not on $\graph{H}$.) We will then define our cost function by
    \[
        C[v] := \Big\lceil \hyperspc \Big\lceil \sum_{c \in \centers} \station{\graph{G},c,\rp}\Big)\Big\rfloor \hyperspc \Big\rfloor[v] \mbox{ for all }v \in \nodes \setminus \trust_\graph{G}.
    \]
    
    In order to bound the cost function later, it helps to do the normalization explicitly.
    Let $\gamma(v) = \frac{1}{|\centers|} \sum_{c\in \centers} \station{\graph{G},c,\rp}[v]$
    and $Z= \sum_{v\in \nodes \setminus \trust_\graph{G}} \gamma(v)$. 
    For $v\in \nodes \setminus \trust_\graph{G}$,
    $C[v] = \gamma(v)/Z$. We will use the fact (which we will prove shortly) that
    \begin{equation}\label{eqn:help}  C[\purchase] \ge \frac{1}{|\centers|}\sum_{c \in \centers}\station{\graph{G},c,\rp}[\purchase].
    \end{equation}
   To  establish~\eqref{eqn:help}, note that 
         \[
        C[\purchase] = \frac{1}{Z}\sum_{v \in \purchase} \gamma(v) = \frac{1}{|\centers|Z}\sum_{v\in\purchase}\sum_{c \in \centers}\station{\graph{G},c,\rp}[v] =  \frac{1}{|\centers|Z}\sum_{c \in \centers}\station{\graph{G},c,\rp}[\purchase].
    \]
    So \eqref{eqn:help} follows from $Z\leq 1$, which follows from the following calculation.
        \[
        Z = \sum_{v \in \nodes \setminus \trust_\graph{G}} \gamma(v) \le \sum_{v \in \nodes}\gamma(v) = \frac{1}{|\centers|}\sum_{c \in \centers}\station{\graph{G},c,\rp}[\nodes] = \frac{|\centers|}{|\centers|} = 1
    \]

    Let $\purchase \subseteq \nodes \setminus \trust_\graph{G}$, let $\graph{H} = (\nodes \cup \spam, \edges_\graph{H}) \in \graph{G}_\purchase$, where $\nodes$ and $\spam$ are disjoint, and let $\centers'$ be the maximal coherent subset of $\centers$ chosen by $\mtermialgot{k}{\rp}$. Then to prove the result, it suffices to show that
    \begin{equation}\label{eqn:mppr-proof-1}
        \mtermi{\graph{H},\centers',\rp}[\spam \cup \purchase] \le \frac{3k}{\rp}C[\purchase].
    \end{equation}
    
    For convenience, we define
    \begin{align*}
        M(\graph{G})[v] &:= \min\{\station{\graph{G},c,\rp}[v]\colon c \in \centers\} \mbox{ for all }v \in \nodes,\\
        M(\graph{H})[v] &:= \min\{\station{\graph{H},c,\rp}[v]\colon c \in \centers\} \mbox{ for all }v \in \nodes \cup \spam.
    \end{align*}
    (Note that $M(\graph{H})[v]$ is defined in terms of $\centers$, not $\centers'$.) We now split into two cases depending on the value of $M(\graph{H})[\nodes]$.
    
    \medskip\noindent\textbf{Case 1:} $\boldsymbol{M(\graph{H})[\nodes] \le 1/3}$. In this case, we will argue that the spammer has had to pay so high a price that the behavior of $\mtermialgot{k}{\rp}$ is irrelevant. We first use the assumption that 
    $\mix_{\graph{G}}\le 1/    (3\rp(3+\log_2 n)) $  to prove that $M(\graph{G})[\nodes] \ge 2/3$. By Lemma~\ref{lem:new-PR-high}, for all $c \in \trust$ and all $v \in \nodes$ with $\station{\graph{G}}[v] \ne 0$,
    \[
        \station{\graph{G},c,\rp}[v] \ge \Big(1-\rp\mix_\graph{G}\big(3-\log_2\station{\graph{G}}[v]\big)\Big)\station{\graph{G}}[v].
    \]
    Summing over all such $v \in \nodes$ and using the fact that $G \in \mixclass{1/3\rp(3+\log_2 n)}$, 
    we obtain
    \begin{equation}\label{eqn:mppr-proof-2}
        M(\graph{G})[\nodes] \ge 1 - \rp\mix_\graph{G}\big(3+H(\station{\graph{G}})\big) \ge 1 - \rp\mix_\graph{G}(3+\log_2 n) \ge 2/3.
    \end{equation}
    
    Now, for each $v \in \nodes \cup \spam$, let $\chi(v)$ be an arbitrary vertex $c \in \centers$ such that $M(\graph{H})[v] = \station{\graph{H},c,\rp}[v]$. For all $c \in \centers$, let $\mathbf{B_c} := \{v \in \nodes\colon \chi(v) = c\}$. Then we have
    \begin{align*}
        M(\graph{G})[\nodes] - M(\graph{H})[\nodes] &\leq \sum_{c \in \centers} \big(\station{\graph{G},c,\rp}[\mathbf{B_c}] - \station{\graph{H},c,\rp}[\mathbf{B_c}] \big)\\
        &=\sum_{c \in \centers} \Big(\big(1 - \station{\graph{G},c,\rp}[\nodes \setminus \mathbf{B_c}]\big) - \big(1 - \station{\graph{H},c,\rp}[(\spam \cup \nodes) \setminus \mathbf{B_c}] \big)\Big)\\
        &\le \sum_{c \in \centers} \Big(\station{\graph{H},c,\rp}[(\nodes \setminus \mathbf{B_c}) \cup \spam] - \station{\graph{G},c,\rp}[(\nodes \setminus \mathbf{B_c}) \setminus \purchase] \Big).
    \end{align*}
    Using Lemma~\ref{lem:PPR}, applied to each term in the sum with $\mathbf{A} = \nodes \setminus \mathbf{B_c}$, it follows that
    \[
        M(\graph{G})[\nodes] - M(\graph{H})[\nodes] \le \sum_{c \in \centers} \frac{1}{\rp}\station{\graph{G},c,\rp}[\purchase].
    \]
    Hence by~\eqref{eqn:help}, we have
    \[
        M(\graph{G})[\nodes] - M(\graph{H})[\nodes] \le \frac{|\centers|}{\rp}C[\purchase] \le \frac{k}{\rp}C[\purchase].
    \]
    Recall that $M(\graph{H})[\nodes] \le 1/3$ by hypothesis, and $M(\graph{G})[\nodes] \ge 2/3$ by~\eqref{eqn:mppr-proof-2}. Thus
    \[
        C[\purchase] \ge \frac{\eps}{3k} \ge \frac{\eps}{3k}\mtermi{\graph{H},\centers',\rp}[\spam\cup\purchase],
    \]
    so~\eqref{eqn:mppr-proof-1} holds as required.
    
    \medskip\noindent\textbf{Case 2:} $\boldsymbol{M(\graph{H})[\nodes] > 1/3}$. In this case, there exists $v \in \nodes$ such that $M(\graph{H})[v] > 0$, so $\centers$ is coherent. Hence $\centers' = \centers$. By the definition of $\mathcal{R}_\textnormal{min}$, it follows that
    \begin{align*}
        \mtermi{\graph{H},\centers',\rp}[\spam \cup \purchase] &= \frac{M(\graph{H})[\spam \cup \purchase]}{M(\graph{H})[\spam \cup \nodes]} \le \frac{M(\graph{H})[\spam \cup \purchase]}{M(\graph{H})[\nodes]} \le 3M(\graph{H})[\spam \cup \purchase]\\
        &\le 3\sum_{c \in \centers} \station{\graph{H},c,\rp}[\spam \cup \purchase].
    \end{align*}
    By Lemma~\ref{lem:PPR}, applied with $\mathbf{A} = \purchase$, and again using~\eqref{eqn:help} it follows that
    \[
        \mtermi{\graph{H},\centers',\rp}[\spam \cup \purchase] \le \frac{3}{\rp}\sum_{c \in \centers}\station{\graph{G},c,\rp}[\purchase] \le \frac{3|\centers|}{\rp}C[\purchase] \le \frac{3k}{\rp}C[\purchase].
    \]
    Thus~\eqref{eqn:mppr-proof-1} holds in all cases, as required.
\end{proof}


%% file: experiments.tex
\section{Experimental Evaluation of PageRanks}
\label{sec:experiments}

In this section, we evaluate the performance of four ranking functions on an actual web graph: Min-PPR, Median-PPR, UPR, and Mean-PPR, the componentwise mean of PPRs. The graph is WEBSPAM-UK2007~\cite{webspam}, which consists of 114,529 hosts,  5,709 of which have been labeled as \emph{trusted} and 344 as \emph{spam}. We evaluate Min-PPR, Median-PPR and Mean-PPR on constituent PPRs based on picking centers at random from among the trusted sites. We consider an empirical evaluation of spam resistance and distortion on a graph that has already been spammed, so we must modify the definitions to suit our needs, as described below. For each ranking function we evaluate:





\paragraph{Spam Resistance.}
We cannot evaluate spam resistance directly using the definition, since we do not know what the pre-spam graph is, but we can evaluate the \textbf{\emph{spam rank}}, or the sum of the ranks assigned to all nodes labeled spam.
In addition to the total ranking for all spammers, we evaluate spam ranks for each ranking function by ordering each node in the graph according to rank and then counting the number of spam sites in each decile. This view of the spam rankings ensures that spam-resistance benefits are well distributed across the ranking vector, rather than potentially reflecting a difference for only a few bad sites. 

Ideally, an optimal spam-resistant ranking function would not inadvertently penalize non-spam sites in an effort to assign low rank to spammers. In order to assess this potential trade-off, we measure the \textbf{\emph{trusted rank}}, or the sum of the ranks of trusted sites, in addition to the spam rank. Analogously to spam rank, we also compute the distribution of trusted sites in each decile of rank.



%
\paragraph{Distortion.}
Since our definition of distortion applies only to strongly connected graphs, for our experiments, we first restrict to the largest strongly connected component of the graph and renormalize each ranking vector. We then compute distortion for UPR, Min-PPR, Mean-PPR and Median-PPR using the stationary distribution on the largest strongly connected component as the reference ranking. We use a value of $\delta = 2$, and we also tested $\delta$ values of 2.5, 3 and 4, but the differences were minimal. 

\paragraph{Stability.}
Since the ranking functions in this paper rely on selecting a set of trusted centers, we would like to examine how sensitive the results are to an accurate assessment of which sites are indeed trusted. Ideally, the spam resistance of a ranking would not swing wildly if a ``wrong'' center made its way into the trusted set. Thus, we evaluate the \emph{stability} of the ranking functions by measuring the changes in rank produced by using maximally bad sites as centers. Specifically, we compute the multiplicative changes in spam rank and trusted rank for a set of experiments run entirely with known spam sites as centers, relative to the same measurements on a set of trusted centers. This metric is not a feature of our theoretical analysis, but provides interesting additional insight in the experimental results. In particular, stability will help distinguish between Min-PPR and Median-PPR, which both perform well under the other three metrics.

\paragraph{Results.}
We find:
\begin{description}
\item[Spam Rank and Trusted Rank:] We find that UPR has spam rank that is 192\% higher than \mbox{Mean-,} Median- and Min-PPR, on average, as expected from the theoretical analysis.  We find that Mean-, Median- and Min-PPR have similar trusted rank once $k$ is at least 3, and that Min-PPR has 26\% higher spam rank than Mean- and Median-PPR on average.
\item[Distortion:] We find that UPR has distortion that is on average 47 times higher than Min- and Median-PPR, and the distortion of Mean-PPR is 15 times higher.
\item[Stability:] Although Min-PPR and Median-PPR are roughly comparable in terms of spam rank, trusted rank and distortion, when we compute both functions using randomly selected spam nodes as centers, Min-~PPR shows greater stability. The Median-PPR spam ranks jump by an average of 68\% and trusted ranks fall by 19\%, while for Min-PPR, the spam rank and trusted rank change by 20\% and -2\%, respectively.  We posit that the instability of Median-PPR is due to the weak closure of the median operator and present evidence to support this hypothesis.
\end{description}

Furthermore, we find that in all cases, picking three centers is enough to achieve most of the benefits of these ranking functions, so the computational cost is only a factor of three higher than that of UPR.  

We conclude that Min-PPR has an attractive combination of spam rank, trusted rank, distortion, stability and computational cost, in accordance with the theoretical results in this paper.

\subsection{Experimental Setup}
\subsubsection{The Web Graph} 
We employed the web graph dataset WEBSPAM-UK2007 \cite{webspam}, which is based on a crawl of web pages in the .UK domain that were labeled by volunteers for research purposes. The dataset contains 114,529 hosts with directed edges between them denoting webpage outlinks. Some of the hosts are labeled as \emph{spam} (344 hosts), or \emph{normal} (5,709 hosts, which we refer to as \emph{trusted}). The largest strongly connected component of the graph has 59,160 nodes, with 134 spam nodes and 3,167 trusted nodes.

\subsubsection{Parameters of the Experiments}
We compute UPR, Mean-PPR, Min-PPR and Median-PPR for
various choices of the reset probability $\epsilon$, and where applicable,
various numbers of trusted centers $k$ and subsets of the trusted sites 
used as centers. 

We ran 50 independent trials for each $k \in \{1,2,\ldots,30\}$ and each
$\epsilon \in \{0.15, 0.05, 0.01, 0.001\}$. For the value of $\epsilon$, 0.15 was selected based on
the historical precedent from \cite{page1999pagerank} of using this reset probability for PageRank, and the smaller values were chosen based on our theoretical results implying that smaller $\epsilon$ yields better outcomes. Of these four, we find that $\epsilon = 0.15$ and $\epsilon = 0.01$ were representative, so these are the ones reported.

For a given $k$, the trusted centers were selected from the trusted sites in the largest strongly connected component. Sites were selected independently at random, according to a distribution of weights corresponding to their relative ranks in the stationary distribution. This reference rank-weighted selection is consistent with Theorem~\ref{thm:minppr-dist}. 

For Min-PPR, Median-PPR and Mean-PPR, each point in Figures~\ref{fig:spam},~\ref{fig:trusted} and~\ref{fig:distortion} reflects the average rank/distortion over all 50 trials for each value of $k$. In Figure~\ref{fig:deciles} and Figure~\ref{fig:trusted-dec}, each decile bar reflects the average count over all trials and all values of $k$. UPR does not select specific centers, and therefore does not vary with $k$.

To run the experiments, we first computed a uniform PageRank ranking vector and PPRs for all
centers for each value of the reset probability $\epsilon$ using the python networkx library \emph{pagerank}()
function. We selected a tolerance for convergence of $10^{-12}$ after testing various tolerance levels, and determining that smaller tolerances did not significantly impact the results.
For rankings with multiple trusted centers, the pointwise averages, minimums and
medians are then computed from the PPRs to produce the Mean-, Min- and Median-PPR
ranking vectors, respectively. Min- and Median-PPR are normalized so that
the total rank of the graph is 1.0.

For the stability experiment, we repeated the spam resistance experiments using centers
selected entirely from the set of spam sites. We ran 10 independent trials for $k=3$ centers. In Section~\ref{sec:exp_results}, we observe that the spam resistance and distortion benefits for Min- and
Median-PPR are achievable with as low as three centers, which informed the selection of $k$ used here.
As with the trusted sites, the spam centers were randomly selected with probabilities determined
by their ranks in the stationary distribution. We then compared the spam resistance metrics --
total spam ranks and total trusted ranks -- to the same metrics for $k=3$ trusted centers, as
described in Section~\ref{sec:stability-res}.

\input{plots}

\subsection{Experimental Results}\label{sec:exp_results}
In this section, we discuss the results of the experiments described in Section~\ref{sec:experiments}.
The results are also presented for each metric and $\epsilon = 0.15, 0.01$ in Figure~\ref{fig:spam} through Figure~\ref{fig:distortion-zoom}.

\subsubsection{Spam Rank Results}
Based on our theoretical results, we would expect each of Mean-,
Min- and Median-PPR to perform better than UPR in terms of spam resistance. Indeed, as shown
in Figure~\ref{fig:spam}, all three of these ranking functions assign significantly lower rank to the
spam sites than UPR across all $k$. 

For $\epsilon = 0.15$, the total rank assigned to spammers by UPR is 0.0021, while the average
spammer ranks over all trials and $k$ were 0.0008 for Min-PPR, 0.0006 for Median-PPR, and 0.0004
for Mean-PPR, for an overall reduction of between 63\% and 74\%. Similarly, for $\epsilon =
0.01$, the total spam rank for UPR is 0.0019, while the total spam ranks average 0.0008 for
Min-PPR, and 0.0007 for Median- and Mean-PPR, for a reduction of between 57\% and 62\%. 
 
Of the three PPR ranking functions, Min-PPR shows the highest spam rank, which is also expected.
Averaged over all $k$ and over $\epsilon \in \{0.15, 0.01\}$, the total spam rank for Min-PPR is
24\% higher than for Mean-PPR. Median-PPR performs the best overall, with 3\% lower average
spam rank than Mean-PPR, with an even better improvement of 15\% seen for $\epsilon=0.01$.

Figure~\ref{fig:deciles} demonstrates that the improvements in spam resistance over
UPR are seen across the ranking vector, and not just in total spam rank. When we order all
nodes of the graph by rank and divide them into deciles, UPR assigns many more spammers to
the highest deciles (8-10) than the other three ranking functions, and fewer spammers to most of
the deciles of lower rank.

Additionally, our experimental results demonstrate that these benefits are achievable with even a very low number of trusted centers -- 2 for Min-PPR and 3 for Median-PPR -- and increasing the centers generally does not significantly increase the spam resistance.

\subsubsection{Trusted Rank Results}
Examining the trusted rank diagnostics, we see that all four ranking functions
are relatively close in trusted site ranking, implying that not much trusted rank is
traded for the increased spam resistance.  In fact, the plots of 
trusted sites in each decile of rank in Figure~\ref{fig:trusted-dec} show a modest shift of trusted 
sites towards the highest decile (10) for Min-, Median- and Mean-PPR compared to UPR. 
For the sum of trusted ranks shown in
Figure~\ref{fig:trusted}, we note that since a PPR assigns at least $\epsilon$ rank to its center, and 
since the centers used are trusted sites, the trusted rank statistic will be inflated for Mean-PPR. 
Therefore, examining the rank of trusted sites apart from the contribution from the center nodes is more informative, and
this statistic is plotted in the green dashed line. This adjustment does not need to be made to Min- and Median-PPR, since
the min and median operations prevent large ranks at the centers.
After applying this correction for Mean-PPR, we also conclude 
from this plot that all four ranking functions show similar overall trusted rank, with Min- and Median-PPR improving by 16-18\% over UPR. 

\subsubsection {Distortion Results}
For distortion, once again all three PPR ranking functions exhibit a significant
improvement over UPR, as shown in Figure~\ref{fig:distortion}. This is particularly interesting with respect to the relative performance of Mean-PPR, which we would expect to have high distortion in the neighborhood of
each trusted center due to the concentration of the reset vector. We find that Mean-PPR's unexpectedly low distortion when compared to UPR is explained by the observation that, for Mean-PPR, the
distortion is typically maximized at the centers, which already tend to have high reference rank because they are trusted.  For UPR, the distortion is maximized at nodes with virtually no reference rank, since all nodes receive an equal share of reset in UPR and therefore have a fairly high minimum rank.

Although better than UPR, the distortion of Mean-PPR is still higher than those of Min- and Median-PPR. This aligns as well with the theoretical results in this paper. Note that for one center, each of Min-, Median- and Mean-PPR is just a PPR, so the three functions exhibit equal distortion, and similarly, Median- and Mean-PPR are equal for $k=2$ since the median of two numbers is their mean.

For $\epsilon = 0.15$, the distortion of UPR is 31,059, while the average distortion over all trials and all $k$ is 274 for Min-PPR, 1,051 for Median-PPR, and 9,482 for Mean-PPR.  For $\epsilon = 0.01$, the distortion of UPR is 15,480, while the average distortion over $k$ is 139 for Min-PPR, 533 for Median-PPR, and 4,647 for Mean-PPR. A zoomed-in view in Figure~\ref{fig:distortion-zoom}  shows that while the distortion is quite low for both Min-PPR and Median-PPR when $k$ is at at least 3,  Min-PPR dominates under this metric for all $k$. Again, we see that increasing the trusted centers beyond $k=3$ does not lower the distortion significantly.

\subsubsection {Stability Results} \label{sec:stability-res}

We compared the spam resistance of each ranking function using three spam centers to our results for three trusted centers. The results in Tables~\ref{tab:stability-.15} and~\ref{tab:stability-.01} summarize average spam and trusted ranks over all trials, along with the overall percentage change produced by running each function with spam centers instead of trusted centers. Min-PPR exhibits greater stability under this variation in centers, with a 12-28\% increase in spam ranks and a 2-3\% decrease in trusted ranks,  versus spam rank increases of 39-105\% and trusted rank decreases of 13-23\% for Median-PPR.


\begin{table*}
  \caption{Stability comparison for 3 centers and $\epsilon = 0.15$}
  \label{tab:stability-.15}
  \begin{tabular}{l|ccc|ccc}
    \toprule
    & & \textbf{spam ranks} & & & \textbf{trusted ranks} & \\
    \midrule
    & \begin{tabular}{c}spam\\ centers\end{tabular} & \begin{tabular}{c}trusted\\ centers\end{tabular} & difference & \begin{tabular}{c}spam\\ centers\end{tabular} & \begin{tabular}{c}trusted\\ centers\end{tabular} & difference \\ 
    \midrule
    Min-PPR & 0.0009 & 0.0007 & 28.5\% & 0.0453 & 0.0466 & -2.8\%\\
    Median-PPR & 0.0011 & 0.0005 & 105.1\% & 0.0450 & 0.0587 & -23.2\%\\
  \bottomrule
\end{tabular}
\end{table*}

\begin{table*}
  \caption{Stability comparison for 3 centers and $\epsilon = 0.01$}
  \label{tab:stability-.01}
  \begin{tabular}{l|ccc|ccc}
    \toprule
    & & \textbf{spam ranks} & & & \textbf{trusted ranks} & \\
    \midrule
    & \begin{tabular}{c}spam\\ centers\end{tabular} & \begin{tabular}{c}trusted\\ centers\end{tabular} & difference & \begin{tabular}{c}spam\\ centers\end{tabular} & \begin{tabular}{c}trusted\\ centers\end{tabular} & difference \\ 
    \midrule
    Min-PPR & 0.0008 & 0.0008 & 12.2\% & 0.0412 & 0.0420 & -2.0\%\\
    Median-PPR & 0.0010 & 0.0007 & 38.5\% & 0.0409 & 0.0469 & -12.8\%\\
  \bottomrule
\end{tabular}
\end{table*}

We hypothesize that the relative instability of Median-PPR is due to the weak closure of the median operator: if the effective $\epsilon$ of the PageRank random walk represented by Median-PPR is large, then random paths will be short, and the PageRank will concentrate around the nodes that receive reset.  Note that nodes throughout the graph may receive reset in a Median-PPR, not just the $k$ PPR centers.  This is because the closure of median only guarantees that there exists \emph{some} (reset vector, $\epsilon$) pair that will yield the Median-PPR as a PageRank.  If the (induced) reset vector for Median-PPR with spam centers is relatively rich in spammers and the $\epsilon$ is large, then this would explain the instability of Median-PPR.

We selected a representative triple of trusted sites and a triple of spam sites to compute the blowup in $\epsilon$.  To compute the effective $\epsilon$, we found the minimum $\epsilon$ with a non-negative reset vector.  As expected, since the min operator is strongly closed, the effective $\epsilon$ for Min-PPR was the original $\epsilon$: either $0.15$ or $0.01$, respectively.    For Median-PPR, the input $\epsilon = 0.15$ blows up to an effective $\epsilon$ of $0.40$ with trusted centers and $0.47$ for spam centers.  When the input $\epsilon = 0.01$, the effective $\epsilon$ for Median-PPR blows up to $0.22$ for trusted centers and $0.36$ for spam centers.  Thus, the effective $\epsilon$ for Median-PPR is quite large.

Table~\ref{tab:epsilons} summarizes the results of measuring the amount of reset that spam nodes receive, for both Median-PPR and Min-PPR, when trusted or spam centers are selected.  Median-PPR gives spam nodes higher reset when spam nodes are selected as centers, which, when combined with a much higher $\epsilon$, explains why Median-PPR is less stable than Min-PPR.

\begin{table*}
  \caption{Sum of spam node reset probabilities}
  \label{tab:epsilons}
  \begin{tabular}{l|ccc|ccc}
    \toprule
    & & \begin{tabular}{c}original\\ $\epsilon=0.15$\end{tabular} & & & \begin{tabular}{c}original\\ $\epsilon=0.01$\end{tabular} & \\
    \midrule
 & \begin{tabular}{c}3 spam\\ centers\end{tabular} & \begin{tabular}{c}3 trusted\\ centers\end{tabular} & difference & \begin{tabular}{c}3 spam\\ centers\end{tabular} & \begin{tabular}{c}3 trusted\\ centers\end{tabular} & difference\\
    \midrule
  Min-PPR & 0.0007 & 0.0004 & 57.2\% & 0.0007 & 0.0004 & 52.6\%\\
 Median-PPR & 0.0011 & 0.0004 & 165.3\% & 0.0010 & 0.0006 & 70.5\%\\
  \bottomrule
\end{tabular}
\end{table*}

\subsection{Conclusion}\label{sec:conclusion}
We tested UPR, Min-PPR, Median-PPR, and Mean-PPR
on a real-world web graph, in order to evaluate their performance under our formalized notions of spam resistance and distortion. The experimental outcomes conform closely with the theoretical results. Namely, UPR has low spam resistance and also suffers from high local distortion. Additionally, while Mean-PPR is significantly more spam resistant than UPR, the distortion around the trusted centers is still high. 

Min-PPR and Median-PPR both exhibit promising performance under both metrics, with Median-PPR performing the best in the spam resistance trials and Min-PPR performing the best with respect to distortion. Finally, we analyzed the stability of both ranking functions under the choice of PPR centers, using a particularly extreme example of a ``wrong'' choice involving only labeled spam centers, and found Min-PPR to be far more stable than Median-PPR.  We show evidence that suggests that the difference in stability may be related to the weaker PageRank closure properties of the median operator compared to min.

Our experiments also showed that the benefits of Min-PPR are achievable even when using as few as two trusted centers. Our results support our theoretical conclusion that Min-PPR enjoys a strong combination of high spam resistance and low distortion with low computational cost.   

%% file: plots.tex
\pgfplotscreateplotcyclelist{mycolors}{
{blue, mark = *,thick},
{purple, mark = diamond,thick},
{teal, mark = square,thick},
{orange, mark = none,thick},
}

\pgfplotscreateplotcyclelist{mycolors_bar}{
{blue, mark = none, fill=blue},
{purple, mark = none, fill=purple},
{teal, mark = none, fill=teal},
{orange, mark = none, fill=orange},
}

\pgfplotscreateplotcyclelist{mycolors_trusted}{
{blue, mark = *,thick},
{purple, mark = diamond,thick},
{teal, mark = square,thick},
{teal, mark = none, dashed, thick},
{orange, mark = none,thick},
}

\begin{figure*}
    \begin{subfigure}{0.49\textwidth}
        \centering
        \begin{tikzpicture}
            \begin{axis} [
                cycle list name = mycolors,
                mark size = 1pt,
                xmin = 0,
                xmax = 30,
                ymin = 0,
                ymax = 0.0025,
                ytick = {0.0005,0.001,0.0015,0.002,0.0025},
                ylabel = {Sum of spam ranks},
                xlabel = {Number of centers},
                yticklabel style={
                    /pgf/number format/fixed,
                    /pgf/number format/precision=5,
                    rotate = 60,
                    anchor = south east
                },
                scaled y ticks=false,
            ]
                \addplot table [col sep=comma] {data/min-85.csv};
                \addlegendentry{Min-PPR}
                \addplot table [col sep=comma] {data/med-85.csv};
                \addlegendentry{Median-PPR}
                \addplot table [col sep=comma] {data/avg-85.csv};
                \addlegendentry{Mean-PPR}
                \addplot table [col sep=comma] {data/pr-85.csv};
                \addlegendentry{UPR}
            \end{axis}
        \end{tikzpicture}
        \caption{$\epsilon=0.15$}
    \end{subfigure}
\hfill
    \begin{subfigure}{0.49\textwidth}
        \centering
        \begin{tikzpicture}
            \begin{axis} [
                cycle list name = mycolors,
                mark size = 1pt,
                xmin = 0,
                xmax = 30,
                ymin = 0,
                ymax = 0.0025,
                ytick = {0.0005,0.001,0.0015,0.002,0.0025},
                ylabel = {Sum of spam ranks},
                xlabel = {Number of centers},
                yticklabel style={
                    /pgf/number format/fixed,
                    /pgf/number format/precision=5,
                    rotate = 60,
                    anchor = south east
                },
                scaled y ticks=false,
            ]
                \addplot table [col sep=comma] {data/min-99.csv};
                \addlegendentry{Min-PPR}
                \addplot table [col sep=comma] {data/med-99.csv};
                \addlegendentry{Median-PPR}
                \addplot table [col sep=comma] {data/avg-99.csv};
                \addlegendentry{Mean-PPR}
                \addplot table [col sep=comma] {data/pr-99.csv};
                \addlegendentry{UPR}
            \end{axis}
        \end{tikzpicture}
        \caption{$\epsilon=0.01$}
    \end{subfigure}
    \caption{Sum of ranks of all labeled spam sites for each ranking function; lower is better. The x-axis indicates the number of trusted centers, $k$, used for Min-PPR, Median-PPR and Mean-PPR. Each point represents the average sum of spam ranks over 50 independent trials. Min-PPR, Median-PPR, and Mean-PPR exhibit substantially lower spam ranks than UPR across all $k$.}
    \label{fig:spam}
\end{figure*}
\begin{figure*}
    \begin{subfigure}{0.49\textwidth}
        \centering
        \begin{tikzpicture}
            \begin{axis} [
                ybar interval = 0.5,
                cycle list name = mycolors_bar,
                bar width = 2.25pt,
                mark size = 1pt,
                xmin = 1,
                xmax = 11,
                ymin = 0,
                ylabel = {Number of spammers},
                xlabel = {Decile of rank},
                xtick=data
            ]
                \addplot table [col sep=comma] {data/dec-min-85.csv};
                \addlegendentry{Min-PPR}
                \addplot table [col sep=comma] {data/dec-med-85.csv};
                \addlegendentry{Median-PPR}
                \addplot table [col sep=comma] {data/dec-avg-85.csv};
                \addlegendentry{Mean-PPR}
                \addplot table [col sep=comma] {data/dec-pr-85.csv};
                \addlegendentry{UPR}
            \end{axis}
        \end{tikzpicture}
        \caption{$\epsilon=0.15$}
    \end{subfigure}
    \hfill
    \begin{subfigure}{0.49\textwidth}
        \centering
        \begin{tikzpicture}
            \begin{axis} [
                ybar interval = 0.5,
                cycle list name = mycolors_bar,
                bar width = 2pt,
                mark size = 1pt,
                xmin = 1,
                xmax = 11,
                ymin = 0,
                ylabel = {Number of spammers},
                xlabel = {Decile of rank},
                xtick=data
            ]
                \addplot table [col sep=comma] {data/dec-min-99.csv};
                \addlegendentry{Min-PPR}
                \addplot table [col sep=comma] {data/dec-med-99.csv};
                \addlegendentry{Median-PPR}
                \addplot table [col sep=comma] {data/dec-avg-99.csv};
                \addlegendentry{Mean-PPR}
                \addplot table [col sep=comma] {data/dec-pr-99.csv};
                \addlegendentry{UPR}
            \end{axis}
        \end{tikzpicture}
        \caption{$\epsilon=0.01$}
    \end{subfigure}
    \caption{Count of spam sites in each decile of rank for each ranking function. The x-axis indicates the decile when each node in the web graph is ordered according to rank, so 10 is the highest decile and a lower count in high deciles is better. Each bar represents the average count of spammers over each $k \in \{1,2,\ldots 30\}$ and 50 independent trials for each $k$. Min-PPR, Median-PPR and Mean-PPR exhibit a general shift of spammers to deciles 3-5 from deciles 8-10 ($\epsilon=0.15$) or 9-10 ($\epsilon=0.01$) compared to UPR.}
    \label{fig:deciles}
\end{figure*}
\begin{figure*}
    \begin{subfigure}{0.49\linewidth}
        \centering
        \begin{tikzpicture}
            \begin{axis} [
                cycle list name = mycolors_trusted,
                mark size = 1pt,
                xmin = 0,
                xmax = 30,
                ymin = 0,
                ymax = 0.3,
                ytick = {0,0.05,0.10,0.15,0.20,0.25},
                ylabel = {Sum of trusted ranks},
                xlabel = {Number of centers},
                yticklabel style={
                    /pgf/number format/fixed,
                    /pgf/number format/precision=5
                },
                scaled y ticks=false,
                legend style={at={(0.975,0.8)}}
            ]
                \addplot table [col sep=comma] {data/trusted-min-85.csv};
                \addlegendentry{Min-PPR}
                \addplot table [col sep=comma] {data/trusted-med-85.csv};
                \addlegendentry{Median-PPR}
                \addplot table [col sep=comma] {data/trusted-avg-85.csv};
                \addlegendentry{Mean-PPR}
                \addplot table [col sep=comma] {data/trusted-avg-85-noself.csv};
                \addlegendentry{Mean-PPR (no self rank)}
                \addplot table [col sep=comma] {data/trusted-pr-85.csv};
                \addlegendentry{UPR}
            \end{axis}
        \end{tikzpicture}
        \caption{$\epsilon=0.15$}
    \end{subfigure}
\hfill
    \begin{subfigure}{0.49\linewidth}
        \centering
        \begin{tikzpicture}
            \begin{axis} [
                cycle list name = mycolors_trusted,
                mark size = 1pt,
                xmin = 0,
                xmax = 30,
                ymin = 0,
                ymax = 0.3,
                ytick = {0,0.05,0.10,0.15,0.20,0.25},
                ylabel = {Sum of trusted ranks},
                xlabel = {Number of centers},
                yticklabel style={
                    /pgf/number format/fixed,
                    /pgf/number format/precision=5
                },
                scaled y ticks=false,
            ]
                \addplot table [col sep=comma] {data/trusted-min-99.csv};
                \addlegendentry{Min-PPR}
                \addplot table [col sep=comma] {data/trusted-med-99.csv};
                \addlegendentry{Median-PPR}
                \addplot table [col sep=comma] {data/trusted-avg-99.csv};
                \addlegendentry{Mean-PPR}
                \addplot table [col sep=comma] {data/trusted-avg-99-noself.csv};
                \addlegendentry{Mean-PPR (no self rank)}
                \addplot table [col sep=comma] {data/trusted-pr-99.csv};
                \addlegendentry{UPR}
            \end{axis}
        \end{tikzpicture}
        \caption{$\epsilon=0.01$}
    \end{subfigure}
    \caption{Sum of ranks of all labeled trusted sites for each ranking function; higher is better. The x-axis indicates the number of trusted centers, $k$, used for Min-PPR, Median-PPR and Mean-PPR. Each point represents the average sum of trusted ranks over 50 independent trials. For Mean-PPR, we also plot the trusted ranks less the contributions from the ``self-ranks'' of each center (the rank of the center in its PPR over $k$). Min-PPR, Median-PPR and the no-self-ranks version of Mean-PPR all exhibit roughly similar trusted rank levels to UPR.}
    \label{fig:trusted}
\end{figure*}
\begin{figure*}
    \begin{subfigure}{0.49\linewidth}
        \centering
        \begin{tikzpicture}
            \begin{axis} [
                ybar interval = 0.5,
                cycle list name = mycolors_bar,
                mark size = 1pt,
                xmin = 1,
                xmax = 11,
                ymin = 0,
                ylabel = {Number of trusted sites},
                xlabel = {Decile of rank},
                xtick=data,
                legend pos = north west
            ]
                \addplot table [col sep=comma] {data/dec-min-85-trust.csv};
                \addlegendentry{Min-PPR}
                \addplot table [col sep=comma] {data/dec-med-85-trust.csv};
                \addlegendentry{Median-PPR}
                \addplot table [col sep=comma] {data/dec-avg-85-trust.csv};
                \addlegendentry{Mean-PPR}
                \addplot table [col sep=comma] {data/dec-pr-85-trust.csv};
                \addlegendentry{UPR}
            \end{axis}
        \end{tikzpicture}
        \caption{$\epsilon=0.15$}
    \end{subfigure}
    \hfill
    \begin{subfigure}{0.49\linewidth}
        \centering
        \begin{tikzpicture}
            \begin{axis} [
                ybar interval = 0.5,
                cycle list name = mycolors_bar,
                mark size = 1pt,
                xmin = 1,
                xmax = 11,
                ymin = 0,
                ylabel = {Number of trusted sites},
                xlabel = {Decile of rank},
                xtick=data,
                legend pos = north west
            ]
                \addplot table [col sep=comma] {data/dec-min-99-trust.csv};
                \addlegendentry{Min-PPR}
                \addplot table [col sep=comma] {data/dec-med-99-trust.csv};
                \addlegendentry{Median-PPR}
                \addplot table [col sep=comma] {data/dec-avg-99-trust.csv};
                \addlegendentry{Mean-PPR}
                \addplot table [col sep=comma] {data/dec-pr-99-trust.csv};
                \addlegendentry{UPR}
            \end{axis}
        \end{tikzpicture}
        \caption{$\epsilon=0.01$}
    \end{subfigure}
    \caption{Count of trusted sites in each decile of rank for each ranking function. The x-axis indicates the decile when each node in the web graph is ordered according to rank, so 10 is the highest decile and a higher count in high deciles is better. Each bar represents the average count of trusted sites over each $k \in \{1,2,\dots 30\}$ and 50 independent trials for each $k$. Min-PPR, Median-PPR and Mean-PPR exhibit a roughly similar distribution to UPR, with a slight increase in trusted sites in the highest decile.}
    \label{fig:trusted-dec}
\end{figure*}
    
\begin{figure*}
    \begin{subfigure}{0.49\linewidth}
        \centering
        \begin{tikzpicture}
            \begin{axis} [
                cycle list name = mycolors,
                mark size = 1pt,
                xmin = 0,
                xmax = 30,
                ymin = 0,
                ytick = {0,5000,10000,15000,20000,25000,30000},
                yticklabels={0,5\,K,10\,K,15\,K,20\,K,25\,K,30\,K},
                ylabel = {Distortion},
                xlabel = {Number of centers},
                yticklabel style={
                    /pgf/number format/fixed,
                    /pgf/number format/precision=5
                },
                scaled y ticks=false,
            ]
                \addplot table [col sep=comma] {data/dist-min-85.csv};
                \addlegendentry{Min-PPR}
                \addplot table [col sep=comma] {data/dist-med-85.csv};
                \addlegendentry{Median-PPR}
                \addplot table [col sep=comma] {data/dist-avg-85.csv};
                \addlegendentry{Mean-PPR}
                \addplot table [col sep=comma] {data/dist-pr-85.csv};
                \addlegendentry{UPR}
            \end{axis}
        \end{tikzpicture}
        \caption{$\epsilon=0.15$}
    \end{subfigure}
    \hfill
    \begin{subfigure}{0.49\linewidth}
        \centering
        \begin{tikzpicture}
            \begin{axis} [
                cycle list name = mycolors,
                mark size = 1pt,
                xmin = 0,
                xmax = 30,
                ymin = 0,
                ymax = 35000,
                ytick = {0,5000,10000,15000,20000,25000,30000},
                yticklabels={0,5\,K,10\,K,15\,K,20\,K,25\,K,30\,K},
                ylabel = {Distortion},
                xlabel = {Number of centers},
                yticklabel style={
                    /pgf/number format/fixed,
                    /pgf/number format/precision=1000
                },
                scaled y ticks=false
            ]
                \addplot table [col sep=comma] {data/dist-min-99.csv};
                \addlegendentry{Min-PPR}
                \addplot table [col sep=comma] {data/dist-med-99.csv};
                \addlegendentry{Median-PPR}
                \addplot table [col sep=comma] {data/dist-avg-99.csv};
                \addlegendentry{Mean-PPR}
                \addplot table [col sep=comma] {data/dist-pr-99.csv};
                \addlegendentry{UPR}
            \end{axis}
        \end{tikzpicture}
        \caption{$\epsilon=0.01$}
    \end{subfigure}
    \caption{Distortion by number of trusted centers $k$ for each ranking function; lower is better. Each point for Min-PPR, Median-PPR and Mean-PPR represents the average distortion over 50 independent trials for each $k$. Min- and Median-PPR exhibit substantially lower distortion than Mean-PPR, which exhibits substantially lower distortion than UPR.}
    \label{fig:distortion}
\end{figure*}
\begin{figure*}
    \begin{subfigure}{0.49\linewidth}
        \centering
        \begin{tikzpicture}
            \begin{axis} [
                cycle list name = mycolors,
                mark size = 1pt,
                xmin = 0,
                xmax = 30,
                ymin = 0,
                ymax = 500,
                ylabel = {Distortion},
                xlabel = {Number of centers},
                yticklabel style={
                    /pgf/number format/fixed,
                    /pgf/number format/precision=5
                },
                scaled y ticks=false
            ]
                \addplot table [col sep=comma] {data/dist-min-85.csv};
                \addlegendentry{Min-PPR}
                \addplot table [col sep=comma] {data/dist-med-85.csv};
                \addlegendentry{Median-PPR}
                \addplot table [col sep=comma] {data/dist-avg-85.csv};
                \addplot table [col sep=comma] {data/dist-pr-85.csv};
            \end{axis}
        \end{tikzpicture}
        \caption{$\epsilon=0.15$}
    \end{subfigure}
    \hfill
    \begin{subfigure}{0.49\linewidth}
        \centering
        \begin{tikzpicture}
            \begin{axis} [
                cycle list name = mycolors,
                mark size = 1pt,
                xmin = 0,
                xmax = 30,
                ymin = 0,
                ymax = 500,
                ylabel = {Distortion},
                xlabel = {Number of centers},
                yticklabel style={
                    /pgf/number format/fixed,
                    /pgf/number format/precision=5
                },
                scaled y ticks=false
            ]
                \addplot table [col sep=comma] {data/dist-min-99.csv};
                \addlegendentry{Min-PPR}
                \addplot table [col sep=comma] {data/dist-med-99.csv};
                \addlegendentry{Median-PPR}
                \addplot table [col sep=comma] {data/dist-avg-99.csv};
                \addplot table [col sep=comma] {data/dist-pr-99.csv};
            \end{axis}
        \end{tikzpicture}
        \caption{$\epsilon=0.01$}
    \end{subfigure}
    \caption{Distortion zoomed in on Min-PPR and Median-PPR. The data is identical to Figure~\ref{fig:distortion}, but the scale of the y-axis is more granular in the lower distortion values, allowing greater visibility of Min-PPR and Median-PPR.}
    \label{fig:distortion-zoom}
\end{figure*}